\documentclass[sigconf,nonacm]{acmart}
\AtBeginDocument{%
  }

\usepackage{soul}
\usepackage{xcolor}
\newcommand{\revA}[1]{\textcolor{black}{#1}}
\newcommand{\revC}[1]{\textcolor{black}{#1}}
\usepackage{tikz}
\usepackage{xcolor}
\usepackage{graphicx}
\usepackage{booktabs}
\usepackage{multirow}
\usepackage{pgfplots}
\pgfplotsset{compat=1.18}
\usetikzlibrary{arrows.meta,decorations.pathreplacing,positioning,fit,backgrounds,shapes.geometric,calc}
\usepgfplotslibrary{fillbetween}

\makeatletter
\renewcommand{\@ACM@checkaffil}{}%
\makeatother
\begin{document}

\title{CutBackdoor: A Circuit Cut Triggered Backdoor Attack on Variational Quantum Algorithms}

\author{Ahatesham Bhuiyan$^{1}$, Hoang Ngo$^{2}$, Cheng Chu$^{3}$, Qian Lou$^{1}$\\
        Lei Jiang$^{4}$, My T. Thai$^{2}$, Mengxin Zheng$^{1}$}
\affiliation{%
  \institution{$^{1}$University of Central Florida \quad
               $^{2}$University of Florida\\
               $^{3}$North Carolina State University \quad
               $^{4}$Indiana University Bloomington}}






\begin{abstract}
  Variational Quantum Algorithms (VQAs) are a leading paradigm for near-term quantum computing, combining parameterized quantum circuits with classical optimization across quantum chemistry, combinatorial optimization, and quantum machine learning. Since real-world VQA deployments routinely require circuits that exceed available hardware capacity, quantum circuit cutting has become an indispensable execution strategy, and pre-trained parameters are increasingly distributed through public repositories, introducing supply-chain security risks that have received little attention. Prior quantum backdoor attacks either introduce detectable circuit modifications or depend on device-specific noise, and none consider circuit cutting as an attack surface. We present CutBackdoor, the first parameter-supply-chain backdoor that uses cut circuit execution from CutQC as the deployment-time trigger against VQAs. Under noisy finite-shot circuit-cut execution, poisoned parameters preserve full-circuit validation performance while substantially increasing cut-path reconstruction error, without any circuit modification. The trigger activates when a resource-limited victim responds to a qubit-capacity mismatch by invoking the cutting workflow, requiring no attacker presence at deployment. We provide a theoretical analysis and empirically validate it across varying shot budgets. Evaluation across multiple VQA benchmarks on IBM quantum backends demonstrates cut-path energy amplification of $1.3\times$ to $2.9\times$ \revA{over clean baselines on the VQE and VQD benchmarks while maintaining small stealthiness error on the full-circuit path. The cut-path gap persists across the evaluated backends and cut placements under matched compilation; Zero-Noise Extrapolation provides only partial mitigation, and the diagonal-cost QAOA benchmark
  delineates the attack's structural boundary.}
\end{abstract}

\maketitle
\pagestyle{empty}          
\thispagestyle{empty}

\section{Introduction}
\label{sec:intro}

Variational Quantum Algorithms (VQAs)\cite{cerezo2021variational} have emerged as one of the leading paradigms for achieving practical quantum advantage in the Noisy Intermediate-Scale Quantum (NISQ) era. VQAs combine shallow parameterized quantum circuits with classical optimization to tolerate current hardware limitations like gate errors and short coherence times, enabling significant progress across diverse high-impact applications. In quantum chemistry, the Variational Quantum Eigensolver (VQE)~\cite{peruzzo2014variational} estimates molecular ground-state energies with high precision, opening pathways to accelerated drug discovery and the design of novel materials~\cite{mcclean2016theory, peruzzo2014variational, kandala2017hardware, mcardle2020quantum, romero2019strategies}. In combinatorial optimization, the Quantum Approximate Optimization Algorithm (QAOA) addresses problems such as combinatorial optimization~\cite{shaydulin2019multistart,farhi2014quantum, crooks2018performance} and credit risk analysis~\cite{egger2020quantum,egger2020credit} that are intractable for classical methods at scale. In quantum machine learning, VQA-based classifiers and generative models have demonstrated competitive accuracy on benchmark tasks~\cite{chu2023iqgan, chu2025lstm, schuld2018supervised}. The versatility and near-term viability of VQAs have made them the central focus of both academic research and industrial quantum computing initiatives.

Scaling VQAs to a large number of qubits is essential for quantum advantage, yet training large-scale VQCs from scratch on NISQ hardware remains fundamentally difficult, making quantum parameter transfer the dominant practical strategy. Each additional qubit doubles the accessible Hilbert space dimension~\cite{9669165}, expanding the representational capacity of VQAs and enabling modeling of larger and more complex problem instances, including larger molecules~\cite{fujii2022deep}, graphs~\cite{augustino2024strategies}, and datasets. In practice, applications of real-world relevance typically demand a substantial qubit count. For example, VQE is unlikely to outperform classical computational chemistry methods without on the order of $\sim$100 qubits~\cite{gonthier2022measurements}, and recent VQA implementations~\cite{9669165,farrell2024scalable,fujii2022deep} commonly target systems exceeding this threshold. A larger qubit count also enriches entanglement structures within VQCs, improving expressive power, reducing optimization loss, and yielding better solution quality. However, scaling VQCs on NISQ hardware introduces a critical training obstacle: noise-induced barren plateaus~\cite{wang2021noise} cause gradients of the cost function to vanish exponentially with circuit size, making random initialization and naive gradient-based optimization impractical for large circuits. Quantum parameter transfer~\cite{skogh2023accelerating,sureshbabu2024parameter} addresses this by reusing parameters optimized for related problem instances to initialize new circuits, accelerating convergence and alleviating flat optimization landscapes. Empirical evidence supports its effectiveness across VQE and QAOA~\cite{shaydulin2023parameter,shaydulin2021qaoakit,galda2021transferability}. Modern quantum software frameworks, including Qiskit~\cite{javadi2024quantum} and PennyLane~\cite{bergholm2018pennylane}, actively support parameter reuse, and pre-trained parameter sets are widely distributed through public repositories and cloud platforms~\cite{shaydulin2021qaoakit}, forming a growing quantum parameter supply chain that closely mirrors the classical deep learning ecosystem.

\begin{figure}
  \resizebox{\columnwidth}{!}{\input{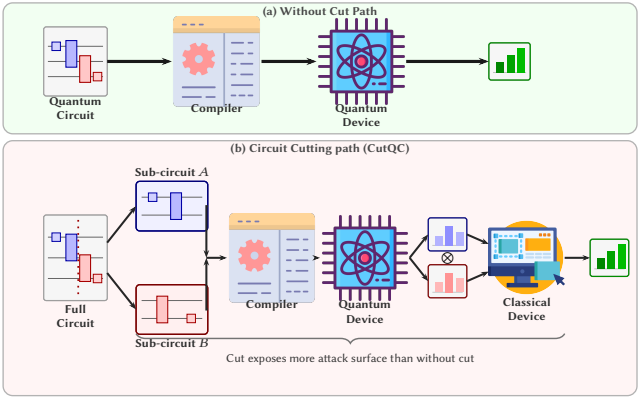}}
  \caption{Circuit cutting (b) exposes a larger attack surface
           than the standard execution path (a).}
  \Description{Two-panel diagram comparing quantum circuit 
  execution paths. Panel (a) shows the standard path: a quantum 
  circuit passes through a compiler to a quantum device. Panel (b) 
  shows the circuit cutting path via CutQC: the full circuit is 
  split into subcircuits A and B, each compiled and executed on a 
  quantum device, then recombined on a classical device. A label 
  notes that the cut path exposes more attack surface.}
\end{figure}

Quantum circuit cutting has emerged as the primary solution to the limited qubit availability on NISQ hardware, a fundamental deployment bottleneck even when parameter transfer alleviates the training difficulty of large-scale VQCs. Practical problem instances regularly require circuits that exceed current hardware capacity~\cite{preskill2018quantum}. Simulating the H$_3^+$ molecule requires at least six qubits, while realistic chemistry and materials-science applications may demand hundreds to thousands~\cite{li2024hybrid,patel2022quest}. Circuit cutting decomposes large circuits into smaller subcircuits that execute independently on constrained hardware, with full-circuit outputs reconstructed through classical postprocessing~\cite{peng2020simulating,tang2021cutqc,piveteau2023circuit,perlin2021quantum}. Two techniques exist: wire cutting~\cite{peng2020simulating,tang2021cutqc}, which interrupts qubit wires via measurement and reinitialization, and gate cutting~\cite{piveteau2023circuit,mitarai2021constructing}, which decomposes multi-qubit gates into local operations with classical post-selection. However, circuit cutting introduces substantial reconstruction overhead~\cite{bechtold2023investigating,lowe2023fast,tang2021cutqc}: the classical and sampling cost of recombining subcircuit outputs via
quasi-probability methods grows exponentially with the number of cuts, making routine validation of cut-based paths prohibitively expensive. Practitioners therefore validate on full-circuit simulators, leaving cut-based execution paths largely unaudited.

Given the central role of circuit cutting in practical VQA deployment, the security implications of cut-based execution are critical yet almost entirely unexplored. Many VQA applications enabled by circuit cutting are security- and safety-sensitive, including molecular energy estimation for drug discovery~\cite{mcclean2016theory,mcardle2020quantum}, materials design~\cite{kandala2017hardware}, and portfolio optimization~\cite{egger2020quantum,egger2020credit}, where errors beyond chemical accuracy thresholds of 1--2 kcal/mol can mislead drug candidate selection or produce financially significant losses~\cite{mcardle2020quantum,romero2019strategies}. The cut-based reconstruction pipeline is inherently more sensitive to adversarial manipulation than full-circuit execution: subcircuit outputs are combined via signed reconstruction coefficients, and even modest systematic biases in these outputs are amplified through classical postprocessing into large errors in the final energy estimate. Despite this, prior research on circuit cutting has focused almost exclusively on reducing reconstruction overhead~\cite{bechtold2023investigating,lowe2023fast,tang2021cutqc}, optimizing cut placement, and extending cutting to new paradigms~\cite{mitarai2021constructing,piveteau2023circuit}. None of these works consider the security implications of adversarially crafted parameters flowing through the cutting pipeline.

Backdoor attacks have emerged as a serious threat to VQCs~\cite{chu2023qdoor,chu2023qtrojan,das2024trojan,chuqnbad}, in which an adversary distributes poisoned parameters that pass standard validation but activate malicious behavior under a specific trigger condition at deployment. Existing approaches rely on triggers tied to circuit structure~\cite{chu2023qtrojan}, compilation processes~\cite{chu2023qdoor,das2024trojan}, or device-specific noise patterns~\cite{chuqnbad}. These attacks either require the attacker to retain control over an external condition at deployment, or depend on a specific compilation setting or noise model that may not generalize across hardware. Exploiting circuit cutting as a trigger is fundamentally harder: the attack must be embedded entirely within the variational parameters, leave no structurally detectable artifact, produce correct results under full-circuit validation while inducing wrong results through the multi-stage CutQC reconstruction, and remain effective across varying backends and compilation configurations.

In this paper, we propose CutBackdoor, the first backdoor attack exploiting quantum circuit cutting as an adversarial trigger against VQAs. The core insight is that full-circuit and cut-based execution are structurally different pipelines, and a carefully crafted parameter set can behave correctly under one while failing systematically under the other. Our contributions are:
\begin{itemize}

\item We identify the noisy finite-shot circuit-cut execution path as a novel adversarial attack surface, and propose CutBackdoor as a dual-objective optimization framework that embeds adversarial behavior into variational parameters without circuit modification, preserving full-circuit stealthiness while substantially increasing cut-path reconstruction error.

\item We formally analyze how the energy deviation between the two execution paths depends on the variational parameters, and show that an adversary can exploit this dependence to maximize cut-path error while keeping full-circuit behavior indistinguishable from clean parameters.

\item We provide a theoretical analysis that upper-bounds the finite-shot gap permitted by cut reconstruction, showing that the attack surface widens with circuit-cutting overhead.


\revA{\item We evaluate CutBackdoor across five VQA benchmarks on IBMQ processors, demonstrating attack effectiveness on the VQE and VQD benchmarks across the evaluated backends and noise profiles, with the diagonal-cost QAOA benchmark marking the attack's structural boundary and Zero-Noise Extrapolation only partially mitigating it.}

\end{itemize}



\section{Background}
\label{sec:relwork}
\subsection{Variational Quantum Algorithms}

\subsubsection{Variational Quantum Algorithms}
Variational Quantum Algorithms (VQAs) are quantum-classical hybrid algorithms designed to solve optimization and simulation problems on near-term NISQ devices. A VQA encodes a problem into a Hamiltonian~$\mathcal{H}$ and employs a parameterized quantum circuit~$U(\boldsymbol{\theta})$, known as the \emph{ansatz}, to prepare a trial quantum state $|\psi(\boldsymbol{\theta})\rangle = U(\boldsymbol{\theta})|\mathbf{0}\rangle$. The objective is to find the optimal parameters that minimize the energy expectation:
\begin{equation}
  E(\boldsymbol{\theta})
  \;=\;
  \langle\psi(\boldsymbol{\theta})|\,\mathcal{H}\,|\psi(\boldsymbol{\theta})\rangle,
\end{equation}
which serves as the cost function. A classical optimizer iteratively updates~$\boldsymbol{\theta}$ based on evaluations of the cost function performed on the quantum device, forming a hybrid quantum-classical feedback loop until convergence. Prominent instances include the Variational Quantum Eigensolver (VQE)~\cite{peruzzo2014variational} for estimating molecular ground-state energies, Variational Quantum Deflation
(VQD)~\cite{higgott2019variational} for computing excited states,
and the Quantum Approximate Optimization Algorithm
(QAOA)~\cite{farhi2014quantum} for combinatorial optimization.

\subsubsection{Parameter Transfer and Quantum Supply Chains}
Training VQA parameters from scratch is challenging due to the barren plateau phenomenon~\cite{larocca2025barren, mcclean2018barren}, which causes cost function gradients to vanish exponentially with qubit count. Quantum parameter transfer~\cite{sureshbabu2024parameter,skogh2023accelerating} addresses this by reusing optimized parameters from related problem instances to initialize new circuits, accelerating convergence and avoiding flat optimization landscapes. This has been demonstrated effectively for QAOA on weighted Max-Cut, VQE across similar molecular Hamiltonians, and general parameterized circuit families. Frameworks such as QAOAKit~\cite{shaydulin2021qaoakit}, Qiskit~\cite{javadi2024quantum}, and PennyLane natively support parameter sharing, and pre-trained parameter sets are increasingly distributed through open-source repositories and quantum cloud platforms. This growing quantum parameter supply chain closely mirrors the classical deep learning ecosystem and inherits the same supply-chain security risks that have proven devastating in the classical domain.

\subsection{Quantum Circuit Cutting}
\label{sec:cutting}

When a target circuit requires more qubits than a device provides, \emph{circuit cutting} decomposes it into subcircuits that fit  within the hardware budget, with full-circuit outputs reconstructed  through classical postprocessing.

\subsubsection{Reconstruction Pipeline.}
\label{reconstruction}
The theoretical foundation is that any quantum channel on a cut wire can be expressed as a linear combination of prepare-and-measure  operations in a chosen basis~\cite{tang2021cutqc}. For \emph{wire cutting}, the Pauli basis $\{I, X, Y, Z\}$ replaces each cut wire with a measurement on the upstream side and a state preparation on the downstream side, across all basis combinations. \emph{Gate cutting}~\cite{piveteau2023circuit} instead decomposes two-qubit  entangling gates into sums of local operations, which is preferable  when wire cutting would break parametric dependencies. CutQC formalized wire cutting into an automated pipeline: it uses Mixed-Integer Programming to find cut locations that minimize postprocessing overhead, generates and executes all required subcircuit variants, and reconstructs the probability distribution of the original circuit. Qiskit's circuit knitting toolbox~\cite{qiskit_cutting2022} extends these ideas to gate cutting and observable estimation.

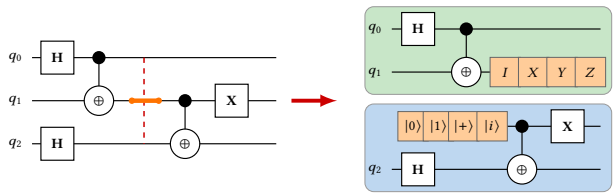
\begin{figure}[htbp]
  \centering
  \resizebox{\columnwidth}{!}{%
    \usetikzlibrary{backgrounds, fit, positioning, arrows.meta}
\definecolor{greenbox}{RGB}{200,230,200}
\definecolor{bluebox}{RGB}{190,215,240}
\definecolor{graybox}{RGB}{215,215,215}

\tikzset{
  gate/.style={rectangle, draw=black, fill=white, line width=0.65pt,
               minimum width=6.5mm, minimum height=6.5mm,
               font=\bfseries\small, inner sep=0pt},
  gbasis/.style={rectangle, draw=black!60, fill=orange!40, line width=0.55pt,
                 minimum width=6.0mm, minimum height=5.8mm,
                 font=\small, inner sep=0pt},
  cdot/.style={circle, fill=black, minimum size=2.6mm, inner sep=0pt},
  cplus/.style={circle, draw=black, fill=white, line width=0.65pt,
                minimum size=5.8mm, inner sep=0pt},
  ql/.style={font=\small, anchor=east},
  w/.style={draw=black, line width=0.65pt},
}

\newcommand{\cnotpair}[3]{%
  \node[cdot] at (#1,#2){};
  \draw[w](#1,#2)--(#1,#3);
  \node[cplus] at (#1,#3){$\oplus$};
}

\def\ws{0.85}  

\begin{tikzpicture}[x=1cm,y=1cm]

\def\Lshift{-0.565}   
\def\ya{\Lshift}
\def\yb{\Lshift-\ws}
\def\yc{\Lshift-2*\ws}

\node[ql] at (0,\ya){$q_0$};
\node[ql] at (0,\yb){$q_1$};
\node[ql] at (0,\yc){$q_2$};

\draw[w](0.08,\ya)--(4.90,\ya);
\draw[w](0.08,\yb)--(4.90,\yb);
\draw[w](0.08,\yc)--(4.90,\yc);

\node[gate] at (0.58,\ya){H};
\node[gate] at (0.58,\yc){H};
\cnotpair{1.40}{\ya}{\yb}
\cnotpair{3.10}{\yb}{\yc}
\node[gate] at (4.02,\yb){X};


\draw[red!80!black, dashed, line width=0.85pt]
  (2.29,{\Lshift+0.001})--(2.29,{\Lshift-2*\ws-0.001});

\draw[orange!90!red, line width=2.5pt, line cap=round]
  (2.05,\yb)--(2.58,\yb);
\filldraw[orange!90!red](2.05,\yb) circle(1.5pt);
\filldraw[orange!90!red](2.58,\yb) circle(1.5pt);


\draw[red!80!black, line width=2.2pt, -latex]
  (5.20,{-1.415})--(6.10,{-1.415});

\def\Rx{7.10}

\def\rya{0}
\def\ryb{-\ws}

\draw[w](\Rx+0.08,\rya)--(\Rx+4.20,\rya);
\draw[w](\Rx+0.08,\ryb)--(\Rx+4.20,\ryb);

\node[ql]  at (\Rx,\rya){$q_0$};
\node[ql]  at (\Rx,\ryb){$q_1$};

\node[gate] at (\Rx+0.60,\rya){H};
\cnotpair{\Rx+1.55}{\rya}{\ryb}

\node[gbasis] at (\Rx+2.32,\ryb){$I$};
\node[gbasis] at (\Rx+2.88,\ryb){$X$};
\node[gbasis] at (\Rx+3.44,\ryb){$Y$};
\node[gbasis] at (\Rx+4.00,\ryb){$Z$};

\begin{scope}[on background layer]
  \fill[greenbox, rounded corners=5pt]
    (\Rx-0.45,\rya+0.44) rectangle (\Rx+4.42,\ryb-0.44);
  \draw[black!40, line width=0.5pt, rounded corners=5pt]
    (\Rx-0.45,\rya+0.44) rectangle (\Rx+4.42,\ryb-0.44);
\end{scope}

\def\gap{0.22}
\def\bya{-2*\ws-\gap}   
\def\byb{-3*\ws-\gap}   

\draw[w](\Rx+1.62,\bya)--(\Rx+4.20,\bya);
\draw[w](\Rx+0.08,\byb)--(\Rx+4.20,\byb);

\node[ql] at (\Rx,\byb){$q_2$};

\node[gbasis] at (\Rx+0.50,\bya){$|0\rangle$};
\node[gbasis] at (\Rx+1.00,\bya){$|1\rangle$};
\node[gbasis] at (\Rx+1.50,\bya){$|{+}\rangle$};
\node[gbasis] at (\Rx+2.06,\bya){$|i\rangle$};

\node[gate] at (\Rx+0.60,\byb){H};

\cnotpair{\Rx+2.65}{\bya}{\byb}

\node[gate] at (\Rx+3.55,\bya){X};

\begin{scope}[on background layer]
  \fill[bluebox, rounded corners=5pt]
    (\Rx-0.45,\bya+0.44) rectangle (\Rx+4.42,\byb-0.44);
  \draw[black!40, line width=0.5pt, rounded corners=5pt]
    (\Rx-0.45,\bya+0.44) rectangle (\Rx+4.42,\byb-0.44);
\end{scope}

\end{tikzpicture}%
  }
  \caption{Circuit cutting example: a 3-qubit circuit is split 
  into two subcircuits via a single wire cut, executed independently, 
  and reconstructed from $4^K$ Pauli-basis branch outputs.}
  \Description{Diagram showing a 3-qubit quantum circuit cut into 
  two subcircuits across the Pauli basis I, X, Y, Z, with the 
  upstream side performing measurements and the downstream side 
  performing state preparations.}
  \label{fig:circuitcut}
\end{figure}

\subsubsection{Subcircuit Probability Distributions}
\label{sec:subcircuit_prob}
Let $U(\boldsymbol{\theta})$ be a parameterized $n$-qubit circuit partitioned into $S$ subcircuits using $K$ wire cuts, producing $M = 4^K$ branch configurations indexed by $b \in [M]$ with signed reconstruction weights $w_b \in \mathbb{R}$. Let $\rho_b^{(j)}(\boldsymbol{\theta})$ denote the exact output state of the $j$-th subcircuit in branch $b$; measuring it yields a
local distribution via Born's rule:
\begin{equation}
    p_{b,x_j}^{(j)}(\boldsymbol{\theta}) =
    \mathrm{Tr}\!\left[
        |x_j\rangle\langle x_j|\,\rho_b^{(j)}(\boldsymbol{\theta})
    \right].
\end{equation}
Since the $S$ subcircuits execute independently, the joint distribution for branch $b$ is
$
    p_b(\boldsymbol{\theta}) =
    \bigotimes_{j=1}^{S} p_b^{(j)}(\boldsymbol{\theta}).
$
The full-circuit output is a weighted sum of $4^K$ Kronecker products: each branch contributes a tensor product of its subcircuit distributions, scaled by its signed weight $w_b$.

\subsubsection{Finite-Shot Measurement}
\label{sec:shots}
Expectation values are estimated by averaging over $N$ independent measurement outcomes, each a \emph{shot}, governed by Born's rule. The empirical estimator is unbiased but carries shot-noise variance scaling as $\mathcal{O}(1/N)$. Under circuit cutting, $N$ shots per subcircuit yield an empirical joint distribution
$
\widehat{p}_b(\boldsymbol{\theta}) =
\bigotimes_{j=1}^{S} \widehat{p}_b^{(j)}(\boldsymbol{\theta}),
$
where each $\widehat{p}_b^{(j)}$ is an unbiased estimator of $p_b^{(j)}$. The pipeline aggregates outputs from $M = 4^K$ branches, each weighted by $w_b$, amplifying the variance by the cutting overhead $\gamma^2 = \left(\sum_b |w_b|\right)^2$ and yielding a variance bound of $\mathcal{O}(\gamma^2/N)$~\cite{peng2020simulating, CarreraVazquez2024}. Since $\gamma$ grows exponentially with $K$, the cut-based estimator is substantially noisier than the full-circuit estimator even when both estimate the same ideal quantity. The $4^K$ sampling overhead makes cut-based validation prohibitively expensive, discouraging practitioners from running it during routine parameter development. These distributions and the variance structure form the building blocks in~\S\ref{sec:setup}.

\subsection{Quantum Compilation}

Quantum compilation translates a high-level quantum circuit into low-level instructions executable on a specific device~\cite{chong2017programming, murali2019noise, javadi2024quantum}. Because physical qubits have limited connectivity, restricted native gate sets, and device-specific noise profiles, compilation critically determines circuit fidelity and depth. \emph{Layout} maps logical qubits to physical qubits on the device coupling graph, with heuristics such as SabreLayout~\cite{li2019tackling} minimizing routing overhead. \emph{Routing} inserts SWAP gates when interacting logical qubits are not physically adjacent; since each SWAP decomposes into three CNOTs, minimizing SWAP count is essential. Qiskit's SabreSwap pass uses a stochastic heuristic, producing circuits whose depth varies across runs. \emph{Translation} decomposes gates into the device's native gate set. \emph{Optimization} eliminates redundant gates and consolidates single-qubit chains. Qiskit exposes four optimization levels (0--3), trading compilation time for gate-count reduction. Different choices of layout, routing, qubit selection, and optimization level produce structurally distinct compiled circuits with different gate counts, depths, and effective noise profiles for the same logical circuit, a sensitivity that directly affects how the attacker's training environment must match the victim's deployment configuration (\S\ref{sec:trigger_analysis}).

\subsection{Noise in NISQ Devices}

NISQ devices~\cite{preskill2018quantum} lack error correction and are limited by several noise sources. \emph{Decoherence} occurs when qubits lose their state due to environmental interactions, characterized by relaxation time $T_1$ and dephasing time $T_2$~\cite{lidar1998decoherence, ithier2005decoherence}. \emph{Gate errors} arise from imperfect control pulses, with two-qubit gates such as CNOT exhibiting typical error rates of $1$--$2\%$ on superconducting devices~\cite{smith2025single}. \emph{Readout errors} corrupt measurement outcomes by mapping $|0\rangle$ to $1$ or $|1\rangle$ to $0$ with non-negligible probability~\cite{o2016scalable}. \emph{Crosstalk} introduces correlated errors when operations on one qubit perturb its neighbours~\cite{sarovar2020detecting, knill2005quantum}. Together these sources constrain the depth and width of reliably executable circuits.

\section{Related Works}

Backdoor attacks were first introduced in classical deep learning, where an adversary poisons training data with a trigger pattern so that the model misclassifies any trigger-stamped input while behaving normally on clean inputs~\cite{gu2017badnets}. Translating this threat to quantum systems has attracted growing attention, though each existing approach targets a different stage of the quantum execution pipeline.Data-poisoning backdoor attacks transplanted from classical neural networks suffer from low success rates on quantum neural networks due to limited input dimensionality, and are eliminated by retraining. Circuit-level attacks~\cite{chu2023qtrojan, das2024trojan, das2023randomized} embed malicious behavior into the ansatz through adversarial gate insertion, but are detectable by inspecting the circuit layout. QDoor~\cite{chu2023qdoor} shifted the attack surface to parameter space via approximate synthesis, producing circuits that
behave correctly before synthesis but maliciously afterwards. A separate line extends backdoor attacks to hybrid classical-quantum
networks~\cite{guo2025backdoor}, and QNBAD~\cite{chuqnbad} crafts parameters that corrupt Zero-Noise Extrapolation under a specific noise model. Circuit cutting has been studied for security only as a defensive primitive~\cite{typaldos2024leveraging, typaldos2025quantum}, using subcircuit submission to obfuscate circuits against untrusted cloud providers.

CutBackdoor departs from all of these by targeting the cutting and reconstruction pipeline rather than the circuit architecture, parameter synthesis, or noise mitigation. Prior attacks require external trigger control at deployment (QTrojan), a specific compilation configuration (QDoor), or a particular device noise model (QNBAD); CutBackdoor instead uses the hardware constraint itself as the trigger, firing whenever the victim's device cannot accommodate the full circuit. The threat is timely because three trends converge in VQA deployment: pre-trained parameters are increasingly shared through public repositories~\cite{shaydulin2021qaoakit, javadi2024quantum},
circuit cutting is unavoidable when applications exceed local hardware capacity~\cite{li2024hybrid, patel2022quest}, and the $4^K$ reconstruction overhead makes routine cut-based validation prohibitively expensive. The resulting validation gap matters most in the domains where VQAs do, including quantum chemistry and portfolio optimization, where small systematic errors can mislead drug-candidate selection or cause significant financial losses~\cite{mcardle2020quantum, egger2020quantum}.

\section{Threat Model}
\label{sec:methodology}

\input{Images/mainfig}

\subsection{Attacker Capabilities and Knowledge}
We consider a threat model consistent with prior work~\cite{chuqnbad, chu2023qdoor}: the attacker has full access to the training process of a variational quantum algorithm and can arbitrarily influence parameter optimization, but for hardware and compilers operates with standard-user permissions and no privileged access. The attacker knows common toolchains (Qiskit, BQSKit, PennyLane) and can simulate end-to-end execution under realistic noise using public simulators and IBMQ hardware. Beyond the circuit, the attacker knows how CutQC is applied at deployment, including how the MIP solver selects cut locations, how subcircuits are generated and assigned parameter subsets, how Pauli bases are enumerated across cut wires, and how the output is reconstructed, and can therefore simulate the cut-based path during training and directly target the reconstructed energy the victim will observe. No interaction with the victim is required at deployment; the attack assumes only that the victim follows the standard CutQC workflow with default subcircuit compilation, under which it transfers across the backends and noise configurations of \S\ref{sec:results} without knowledge of the victim's device, calibration, or mitigation strategy. Sensitivity to non-default compilation is analyzed in \S\ref{sec:trigger_analysis}. The trigger is the structural mismatch between circuit size and available hardware, which the attacker anticipates confidently since pre-trained parameters are most valuable precisely when the circuit exceeds local hardware capacity.

\subsection{Attacker Goals}
The attacker has two simultaneous goals. The primary goal is \textit{attack effectiveness}: when the victim applies circuit cutting, the reconstructed energy diverges from the full-circuit value the victim validated, a deviation that reads as ordinary cutting overhead and device noise. The secondary goal is \textit{attack stealthiness}: on a full-circuit simulator, the same parameters stay close to the clean baseline and pass routine validation. \revA{The adversarial training procedure satisfies both tensioned goals by penalizing full-circuit energy error while pushing the reconstructed energy toward the wrong target on the cut-based path. A victim validating only on a full-circuit simulator observes accurate results and has no indication the parameters are compromised; the backdoor activates silently at deployment when hardware constraints force circuit cutting.}

\subsection{Victim Assumptions}

The victim is a legitimate quantum practitioner who downloads pre-trained VQA parameters from a public source for a computational task such as molecular energy estimation or combinatorial optimization. They validate the parameters on a full-circuit simulator before deployment, which is the standard and computationally practical approach. At deployment, the victim's hardware has fewer qubits than the circuit requires, so circuit cutting is applied automatically. The victim has no reason to suspect the parameters are malicious, since they passed full-circuit validation, and has no practical means to perform exhaustive cut-based validation due to its prohibitive sampling overhead.

\section{CutBackdoor}
\label{cut}

\subsection{Trigger Mechanism}
The trigger in CutBackdoor is not an external input, server configuration, or device-specific noise condition, but a structural property of the victim's deployment: when the circuit requires $n$ qubits and the device supports only $w < n$, circuit cutting becomes the only viable path. This constraint is outside the victim's control, making the trigger automatic and undetectable under standard validation. Figure~\ref{fig:threat_model} illustrates the scenario: the attacker uploads poisoned parameters to a public repository; the
victim validates them on a full-circuit simulator (panel b), but at deployment hardware constraints force the CutQC pipeline (panels c, d), where the backdoor activates.

\subsubsection{CutQC as the Trigger}
CutBackdoor targets CutQC~\cite{tang2021cutqc}, the first and most widely adopted automated cutting pipeline. When hardware cannot accommodate the full circuit, CutQC uses Mixed-Integer Programming to place cuts, executes the subcircuit variants independently, and reconstructs the output through classical postprocessing. Because this pipeline is deterministic given the circuit and hardware, the attacker can simulate it during training and craft parameters whose finite-shot reconstruction deviates systematically from the full-circuit result.

\subsubsection{Variance Asymmetry Between Execution Paths.}
The trigger exploits a finite-shot asymmetry between two execution modes of the same circuit. Under full-circuit execution, the unitary is applied coherently and the energy is estimated by direct measurement. Under CutQC execution, the circuit is partitioned into subcircuits executed independently across all Pauli bases, and the output is reconstructed by combining branch outputs through signed weights. This multi-stage reconstruction amplifies sampling fluctuations far more than direct measurement, growing exponentially in the number of cuts. The two estimators share the same ideal expectation, but their finite-shot realizations diverge: crafted parameters can preserve full-circuit accuracy while steering the cut-path estimator into high-variance regions where reconstruction error is amplified, without any circuit modification. We bound this divergence in \S\ref{sec:bound}.

\subsubsection{Stealth.}

The trigger is stealthy for three reasons. First, the trigger condition is indistinguishable from a normal deployment scenario, since any practitioner who responds to a qubit-capacity mismatch by invoking CutQC activates it. Second, the backdoored parameters pass standard full-circuit validation, which is the only computationally affordable validation approach available to the victim. Third, when the attack fires, the elevated cut-path energy is naturally attributed to reconstruction noise or cutting overhead rather than parameter poisoning. Unlike QTrojan~\cite{chu2023qtrojan}, where injected gates are structurally detectable, or QNBAD~\cite{chuqnbad}, where the attack is device-specific, CutBackdoor leaves no observable artifact in the circuit architecture and persists across the evaluated backends and cut placements under matched compilation.

\begin{figure}[t]
  \centering

\definecolor{hfill}{RGB}{80,80,200}
\definecolor{rzfill}{RGB}{180,90,200}
\definecolor{sqxfill}{RGB}{60,160,210}
\definecolor{cb}{RGB}{65,130,200}

\tikzset{
  wr/.style    = {draw=black, line width=0.5pt},
  Hg/.style    = {rectangle, draw=hfill!70!black, fill=hfill, text=white,
                  font=\bfseries\fontsize{5}{5}\selectfont,
                  minimum width=4.5mm, minimum height=3.8mm, inner sep=0pt},
  Rg/.style    = {rectangle, draw=rzfill!70!black, fill=rzfill, text=white,
                  font=\bfseries\fontsize{3.8}{4}\selectfont,
                  minimum width=5.2mm, minimum height=3.8mm, inner sep=0pt, align=center},
  Sg/.style    = {rectangle, draw=sqxfill!70!black, fill=sqxfill, text=white,
                  font=\bfseries\fontsize{3.8}{4}\selectfont,
                  minimum width=5.2mm, minimum height=3.8mm, inner sep=0pt, align=center},
  dot/.style   = {circle, fill=cb, draw=cb, minimum size=1mm, inner sep=0pt},
  ql/.style    = {font=\fontsize{5.5}{6}\selectfont, anchor=east, text=black},
  sl/.style    = {font=\fontsize{5.5}{6}\selectfont\sffamily, anchor=north, text=black},
  qn/.style    = {circle, draw=cb!80!black, fill=cb, text=white,
                  font=\bfseries\fontsize{4.5}{4.5}\selectfont,
                  minimum size=4.2mm, inner sep=0pt},
  qe/.style    = {draw=cb, line width=1.2pt, line cap=round},
}

\newcommand{\cnot}[3]{%
  \node[dot] at (#1,#2){};
  \draw[cb, line width=0.55pt](#1,#2)--(#1,#3);
  \draw[cb, fill=white, line width=0.55pt](#1,#3) circle(0.065);
  \draw[cb, line width=0.55pt](#1-0.065,#3)--(#1+0.065,#3);
  \draw[cb, line width=0.55pt](#1,#3-0.065)--(#1,#3+0.065);
}

\def\p{0.36}

\begin{tikzpicture}[x=1cm,y=1cm]


\begin{scope}
  \coordinate(N4) at (0.35, 0);
  \coordinate(N5) at (1.15, 0);
  \coordinate(N6) at (1.95, 0);
  \coordinate(N3) at (1.15,-0.68);
  \coordinate(N0) at (0.35,-1.36);
  \coordinate(N1) at (1.15,-1.36);
  \coordinate(N2) at (1.95,-1.36);
  \draw[qe](N4)--(N5)--(N6);
  \draw[qe](N5)--(N3)--(N1);
  \draw[qe](N0)--(N1)--(N2);
  \foreach \q/\c in {4/N4,5/N5,6/N6,3/N3,0/N0,1/N1,2/N2}
    \node[qn] at (\c) {\q};
  \node[sl] at (1.15,-1.72) {(a) IBMQ Oslo coupling map};
\end{scope}

\begin{scope}[shift={(2.80,0)}]
  \foreach \q/\k in {0/0,1/1,2/2,3/3,4/4}
    \node[ql] at (0,-\k*\p) {$q_{\q}$};
  \foreach \k in {0,1,2,3,4}
    \draw[wr](0.04,-\k*\p)--(3.35,-\k*\p);
  \node[Hg] at (0.46,0){H};
  \cnot{1.00}{0}{-\p}
  \cnot{1.60}{-\p}{-2*\p}
  \cnot{2.20}{0}{-3*\p}
  \cnot{2.80}{-2*\p}{-4*\p}

  \node[sl] at (1.65,-4*\p-0.28) {(b) Original circuit};
\end{scope}

\def\rowB{-2.72}   

\begin{scope}[shift={(0,\rowB)}]
  \def\WL{6.4}     
  \node[ql] at (0,0)        {$q_0{\to}0$};
  \node[ql] at (0,-\p)      {$q_1{\to}1$};
  \node[ql] at (0,-2*\p)    {$q_2{\to}2$};
  \node[ql] at (0,-3*\p)    {$q_3{\to}3$};
  \node[ql] at (0,-4*\p)    {$q_4{\to}4$};
  \node[font=\fontsize{5.5}{6}\selectfont\itshape,anchor=east]
        at (0,-5*\p) {$\mathit{anc}{\to}5$};
  \foreach \k in {0,1,2,3,4,5}
    \draw[wr](0.04,-\k*\p)--(\WL,-\k*\p);

  \node[Rg] at (0.3,0)  {$R_z$\\$\!\frac{\pi}{2}\!$};
  \node[Sg] at (0.8,0)  {$\!\sqrt{X}\!$};
  \node[Rg] at (1.3,0)  {$R_z$\\$\!\frac{\pi}{2}\!$};

  \def\startX{1.8}
  \def\stepX{0.28}
  \def\cnt{0}
  \newcommand{\addcnot}[2]{%
    \pgfmathsetmacro{\x}{\startX + \cnt*\stepX}
    \cnot{\x}{#1}{#2}
    \pgfmathsetmacro{\cnt}{\cnt+1}
    \xdef\cnt{\cnt}
  }
  \addcnot{0}{-\p}           
  \addcnot{-\p}{-2*\p}           
  \addcnot{0}{-\p}       
  \addcnot{-\p}{0}           
  \addcnot{0}{-\p}      
  \addcnot{-\p}{-3*\p}       
  \addcnot{-2*\p}{-\p}           
  \addcnot{-\p}{-2*\p}     
  \addcnot{-2*\p}{-\p}        
  \addcnot{-\p}{-3*\p}     
  \addcnot{-3*\p}{-\p}     
  \addcnot{-\p}{-3*\p}     
  \addcnot{-3*\p}{-5*\p}     
  \addcnot{-5*\p}{-3*\p}     
  \addcnot{-3*\p}{-5*\p}     
  \addcnot{-5*\p}{-4*\p}     

  \node[sl] at (3.1,-5*\p-0.28)
    {(c) Trivial layout $+$ Basic route \enspace (2Q:\,16,\ depth:\,19)};
\end{scope}

\def\rowC{-5.22}

\begin{scope}[shift={(0,\rowC)}]
  \def\WL{5.10}
  \node[ql] at (0,0)       {$q_3{\to}0$};
  \node[ql] at (0,-\p)     {$q_0{\to}1$};
  \node[ql] at (0,-2*\p)   {$q_1{\to}3$};
  \node[ql] at (0,-3*\p)   {$q_4{\to}4$};
  \node[ql] at (0,-4*\p)   {$q_2{\to}5$};
  \foreach \k in {0,1,2,3,4}
    \draw[wr](0.04,-\k*\p)--(\WL,-\k*\p);
  \node[Rg] at (1.18,-\p)  {$R_z$\\$\!\frac{\pi}{2}\!$};
  \node[Sg] at (1.68,-\p)  {$\!\sqrt{X}\!$};
  \node[Rg] at (2.18,-\p)  {$R_z$\\$\!\frac{\pi}{2}\!$};
  \cnot{2.95}{-\p}{-2*\p}
  \cnot{3.55}{-\p}{0}
  \cnot{3.55}{-2*\p}{-4*\p}
  \cnot{4.25}{-3*\p}{-4*\p}

  \node[sl] at (2.55,-4*\p-0.28)
    {(d) Sabre layout $+$ Sabre route \enspace (2Q:\,4,\ depth:\,6)};
\end{scope}

\end{tikzpicture}
  \caption{Comparison of transpilation strategies on the 
  \texttt{IBMQ\_Oslo} device. (a) Coupling map of the 7-qubit 
  heavy-hex topology. (b) Original logical circuit. (c) Trivial 
  initial layout with basic routing yields 16 two-qubit gates 
  at depth 19. (d) Sabre layout with Sabre routing yields only 
  4 two-qubit gates at depth 6.}
  \Description{Four-panel diagram. Panel (a) shows the heavy-hex 
  coupling graph of IBMQ Oslo with seven nodes. Panel (b) shows 
  the original logical circuit. Panel (c) shows the same 
  circuit compiled with trivial layout and basic routing, 
  resulting in many SWAP insertions. Panel (d) shows the 
  circuit compiled with Sabre layout and routing, resulting in 
  far fewer two-qubit gates and shallower depth.}
  \label{fig:trigger}
\end{figure}
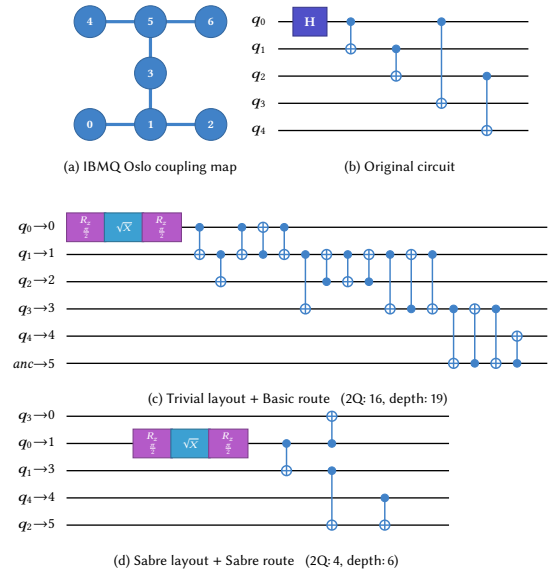

\subsection{Trigger Condition Analysis}
\label{sec:trigger_analysis}
The CutBackdoor trigger activates through the CutQC subcircuit execution path, and its effectiveness depends on whether the victim's subcircuit compilation matches the attacker's training environment. During adversarial training, the attacker simulates the full CutQC pipeline under a fixed subcircuit compilation, embedding adversarial biases calibrated to a specific compiled topology. If the victim compiles subcircuits differently, the topology changes and the adversarial bias may fail to propagate through reconstruction.

Figure~\ref{fig:trigger} shows how much the same logical subcircuit can differ under two compilation strategies on the IBMQ\_Oslo coupling map. Under a trivial layout with BasicSwap routing, sequential qubit assignment forces numerous SWAPs along non-adjacent paths, yielding 16 two-qubit gates at depth 19. Under SabreLayout with SabreSwap, the layout heuristic places interacting qubits on adjacent nodes and the same subcircuit compiles to only 4 two-qubit gates at depth 6. This fourfold reduction in gate count and threefold
reduction in depth alter both noise accumulation and qubit ordering, shifting the assembled output distribution that feeds the energy computation. The adversarial bias calibrated during training no longer maps onto this altered reconstruction, and the attack energy may not reach its target. This sensitivity affects only the cut-based path: full-circuit validation uses the statevector path, independent of subcircuit compilation.

In practice, however, these deviations represent non-standard usage outside the default CutQC workflow. Practitioners using CutQC for automatic cutting naturally adopt low-overhead subcircuit compilation, since aggressive routing is unnecessary when subcircuits are already small. Any victim operating within the standard CutQC pipeline therefore uses a configuration consistent with the attacker's training environment, so the trigger fires reliably without
attacker intervention.


\subsection{Attack Overview}
CutBackdoor embeds adversarial behavior into the variational parameter set $\boldsymbol{\theta}$ of a parameterized quantum circuit such that the same parameters produce two distinct finite-shot behaviors depending on the execution path. Under full-circuit execution, the circuit produces an energy estimate close to the true ground-state value. Under noisy finite-shot CutQC-based execution, the cut-path estimator is steered toward a target value $E_{\mathrm{wrong}}$ chosen by the attacker. The attack requires no modification to the circuit architecture, injects no additional gates, and leaves no structurally detectable artifact.

\subsection{Setup and Notation}
\label{sec:setup}

Let $U(\boldsymbol{\theta})$ be a parameterized $n$-qubit circuit with parameters $\boldsymbol{\theta} \in \mathbb{R}^d$, and let $\mathcal{H}$ be a target observable with diagonal elements $h_x$ in the computational basis. We formalize the energy-level quantities used throughout the attack
analysis across two evaluation settings. The underlying subcircuit probability distributions $p_b^{(j)} (\boldsymbol{\theta})$ and their empirical counterparts $\widehat{p}_b^{(j)}(\boldsymbol{\theta})$ are defined in \S\ref{sec:subcircuit_prob}.

\noindent\textbf{Ideal case.}
The exact full-circuit energy is:
\begin{equation}
    E_{\mathrm{full}}(\boldsymbol{\theta})
    :=
    \mathrm{Tr}\!\left[
        \mathcal{H}\,\rho_{\mathrm{full}}
        (\boldsymbol{\theta})
    \right],
\end{equation}
where $\rho_{\mathrm{full}}(\boldsymbol{\theta})$ is the exact output state of the uncut circuit. For the cut-based path, the exact branch energy for branch $b \in [M]$ is:
\begin{equation}
    \mu_b(\boldsymbol{\theta})
    :=
    \sum_x p_{b,x}(\boldsymbol{\theta})\,h_x,
\end{equation}
where $p_{b,x}(\boldsymbol{\theta})$ is the joint probability over the full-system bitstring $x$, obtained from the tensor product of subcircuit
distributions as defined in \S\ref{sec:subcircuit_prob}. The ideal cut reconstruction is then:
\begin{equation}
    E_{\mathrm{cut}}(\boldsymbol{\theta})
    :=
    \sum_{b=1}^{M} w_b\,\mu_b(\boldsymbol{\theta}),
    \label{eq:ecut_ideal}
\end{equation} and under exact cutting, $E_{\mathrm{cut}}(\boldsymbol{\theta}) = E_{\mathrm{full}}(\boldsymbol{\theta})$~\cite{peng2020simulating}.

\noindent\textbf{Practical case (finite shots).}
As introduced in \S\ref{sec:shots}, in practice
$N$ independent shots are allocated to any executed
circuit or subcircuit. The empirical branch energy
estimate for branch $b$ is:
\begin{equation}
    \widehat{\mu}_b(\boldsymbol{\theta})
    :=
    \sum_x
    \widehat{p}_{b,x}(\boldsymbol{\theta})\,h_x,
\end{equation}
and the finite-shot estimators for the two execution
paths are:
\begin{align}
    \widehat{E}_{\mathrm{full}}(\boldsymbol{\theta})
    &:=
    \sum_{x}
    \widehat{p}_{\mathrm{full},x}
    (\boldsymbol{\theta})\,h_x,
    \label{eq:efull_empirical}
    \\
    \widehat{E}_{\mathrm{cut}}(\boldsymbol{\theta})
    &:=
    \sum_{b=1}^{M}
    w_b\,\widehat{\mu}_b(\boldsymbol{\theta}).
    \label{eq:ecut_empirical}
\end{align}
Both estimators are unbiased:
$\mathbb{E}[\widehat{E}_{\mathrm{full}}
(\boldsymbol{\theta})] =
E_{\mathrm{full}}(\boldsymbol{\theta})$
and
$\mathbb{E}[\widehat{E}_{\mathrm{cut}}
(\boldsymbol{\theta})] =
E_{\mathrm{cut}}(\boldsymbol{\theta})$.
However, as established in \S\ref{sec:shots},
the variance of $\widehat{E}_{\mathrm{cut}}$
is amplified by $\gamma^2 =
\left(\sum_b |w_b|\right)^2$ relative to
$\widehat{E}_{\mathrm{full}}$, a structural
asymmetry that the attack directly exploits.

\subsection{Exploiting the Cutting-Induced Attack Surface}
\label{sec:finite}
In practice, the ideal expectations $E_{\mathrm{full}}(\boldsymbol{\theta})$ and $E_{\mathrm{cut}}(\boldsymbol{\theta})$ are inaccessible; the victim and the execution environment observe only the finite-shot estimators $\widehat{E}_{\mathrm{full}}(\boldsymbol{\theta})$ and $\widehat{E}_{\mathrm{cut}}(\boldsymbol{\theta})$ of \S\ref{sec:setup}. Both
are unbiased estimators of the same ideal quantity, yet their finite-shot realizations diverge due to measurement randomness (\S\ref{sec:shots}). We define the \emph{finite-shot evaluation gap} as:
\begin{equation}
    \Delta(\boldsymbol{\theta})
    :=
    \left|
        \widehat{E}_{\mathrm{cut}}(\boldsymbol{\theta})
        -
        \widehat{E}_{\mathrm{full}}(\boldsymbol{\theta})
    \right|.
    \label{eq:gap}
\end{equation}
This gap is governed by the variances of the two estimators. For the full
circuit, the standard Monte Carlo variance is bounded by the spread of the
observable over the exact state distribution:
\begin{equation}
    \mathrm{Var}\!\left(
        \widehat{E}_{\mathrm{full}}(\boldsymbol{\theta})
    \right)
    =
    \frac{1}{N}
    \left(
        \sum_x
        p_{\mathrm{full},x}(\boldsymbol{\theta})\,h_x^2
        -
        E_{\mathrm{full}}(\boldsymbol{\theta})^2
    \right).
\end{equation}
The variance of $\widehat{E}_{\mathrm{cut}}(\boldsymbol{\theta})$ is
fundamentally different. As established in \S\ref{sec:shots}, the empirical joint distribution is built from the tensor product of independent subcircuit measurements and recombined with signed weights $w_b$, so statistical errors compound and the variance is amplified by the cutting overhead $\gamma^2 = \left(\sum_b |w_b|\right)^2$, yielding the looser bound $\mathcal{O}(\gamma^2/N)$~\cite{peng2020simulating, CarreraVazquez2024}.
Both variances are explicit functions of the distributions $p_{\mathrm{full}}(\boldsymbol{\theta})$ and $p_b^{(j)}(\boldsymbol{\theta})$ (\S\ref{sec:subcircuit_prob}), and are therefore highly sensitive to the location of $\boldsymbol{\theta}$ in the parameter landscape.

This $\boldsymbol{\theta}$-dependence creates an exploitable attack surface. An adversary treats the gap $\Delta(\boldsymbol{\theta})$ as an objective: by tuning $\boldsymbol{\theta}$, the attacker steers subcircuits into regions where the local variance of $\widehat{p}_b^{(j)}(\boldsymbol{\theta})$ is large, and the cutting pipeline amplifies these fluctuations through tensor products and signed weights, inflating the sampling error of $\widehat{E}_{\mathrm{cut}}(\boldsymbol{\theta})$ within the bound of
Theorem~\ref{thm:attack_loss_bound} while the un-amplified $\widehat{E}_{\mathrm{full}}(\boldsymbol{\theta})$ stays comparatively stable.

\subsection{Theoretical Bound on the Finite-Shot
Evaluation Gap}
\label{sec:bound}
Building on the evaluation gap $\Delta(\boldsymbol{\theta})$ defined in \S\ref{sec:finite}, we now establish a high-probability upper bound on how far the cut-path estimator can deviate from the full-circuit estimator under finite shots. The bound is one-sided: it limits the worst-case gap an attacker can achieve, but does not by itself prove that any particular attack realizes the bound. The empirical attack effectiveness within this bound is demonstrated in \S\ref{sec:results}.

\begin{theorem}[Finite-shot bound on the evaluation gap]
\label{thm:attack_loss_bound}
Let $\mathcal{H}$ be an observable with eigenvalues bounded in $[-1, 1]$ (e.g., a Pauli string). Suppose each of the $M = 4^K$ cut branches and the full circuit are estimated using the same number of independent measurement shots $N$. Then, for any confidence parameter $\delta \in (0,1)$, with probability at least $1-\delta$, the finite-shot evaluation gap satisfies:
\begin{equation}
    \Delta(\boldsymbol{\theta})
    \leq
    (\gamma+1)
    \sqrt{\frac{2\log(4M/\delta)}{N}},
    \label{eq:attack_loss_bound}
\end{equation}
where $\gamma := \sum_{b=1}^{M} |w_b|$.
\end{theorem}

\begin{proof}
By the triangle inequality and the fact that $E_{\mathrm{full}}(\boldsymbol{\theta}) = E_{\mathrm{cut}}(\boldsymbol{\theta})$
from~\eqref{eq:ecut_ideal}:
\begin{align}
    \Delta(\boldsymbol{\theta})
    \leq
    \bigl|
        \widehat{E}_{\mathrm{cut}}(\boldsymbol{\theta})
        - E_{\mathrm{cut}}(\boldsymbol{\theta})
    \bigr|
    +
    \bigl|
        \widehat{E}_{\mathrm{full}}(\boldsymbol{\theta})
        - E_{\mathrm{full}}(\boldsymbol{\theta})
    \bigr|.
    \label{eq:triangle_attack}
\end{align}

Because the observable spectrum is bounded in $[-1,1]$, the energy evaluated from any single measurement shot falls within $[-1,1]$. The
empirical branch estimator $\widehat{\mu}_b(\boldsymbol{\theta})$, defined
in \S\ref{sec:setup}, is the sample mean of $N$ such independent single shot measurements. By Hoeffding's inequality applied to each branch
$b$ and a union bound over all $b \in [M]$, with probability at least $1-\delta/2$:
\begin{equation}
    |\widehat{\mu}_b(\boldsymbol{\theta})
    - \mu_b(\boldsymbol{\theta})|
    \leq
    \sqrt{\frac{2\log(4M/\delta)}{N}}
    \qquad \text{for all } b = 1,\dots,M.
\end{equation}
The total sampling error of the cut estimator
is thus bounded by:
\begin{align}
    \bigl|
        \widehat{E}_{\mathrm{cut}}(\boldsymbol{\theta})
        - E_{\mathrm{cut}}(\boldsymbol{\theta})
    \bigr|
    &=
    \left|
        \sum_{b=1}^{M} w_b
        \bigl(
            \widehat{\mu}_b(\boldsymbol{\theta})
            - \mu_b(\boldsymbol{\theta})
        \bigr)
    \right|
    \nonumber \\
    &\leq
    \sum_{b=1}^{M} |w_b|\,
    |\widehat{\mu}_b(\boldsymbol{\theta})
    - \mu_b(\boldsymbol{\theta})|
    \nonumber \\
    &\leq
    \gamma
    \sqrt{\frac{2\log(4M/\delta)}{N}}.
    \label{eq:cut_bound}
\end{align}

Similarly, $\widehat{E}_{\mathrm{full}} (\boldsymbol{\theta})$ is the sample mean of $N$ independent bounded single-shot measurements.
By Hoeffding's inequality and using $M = 4^K \geq 1$, with probability at least $1-\delta/2$:
\begin{equation}
    \bigl|
        \widehat{E}_{\mathrm{full}}(\boldsymbol{\theta})
        - E_{\mathrm{full}}(\boldsymbol{\theta})
    \bigr|
    \leq
    \sqrt{\frac{2\log(4M/\delta)}{N}}.
    \label{eq:full_bound}
\end{equation}

Applying a final union bound,~\eqref{eq:cut_bound} and~\eqref{eq:full_bound} hold simultaneously with probability at least $1-\delta$. Substituting into~\eqref{eq:triangle_attack}
yields the combined factor $(\gamma+1)$.
\end{proof}

\noindent\textbf{Remark.}
Theorem~\ref{thm:attack_loss_bound} bounds the realized gap by
$\mathcal{O}(\gamma/\sqrt{N})$, and since $\gamma$ grows exponentially with the number of cuts $K$, more cuts admit a larger worst-case deviation. The bound is one-sided: it caps the attack surface without guaranteeing any attack saturates it. The reconstruction overhead that discourages cut-path validation thus also widens this surface. We confirm the $K$-dependence empirically in \S\ref{sec:results}, where the realized gap grows with $K$ yet stays well below the ceiling.

\subsection{Backdoor Training Objective}
\label{sec:loss}

We formulate CutBackdoor as a multi-task optimization problem. The attacker trains a poisoned parameter set $\boldsymbol{\theta}^*$ to simultaneously minimize the full-circuit energy and maximize the
finite-shot evaluation gap $\Delta(\boldsymbol{\theta})$ defined in \S\ref{sec:finite}. The overall training loss is:
\begin{align}
    \mathcal{L}(\boldsymbol{\theta})
    =
    \underbrace{
        \left|
            \widehat{E}_{\mathrm{full}}
            (\boldsymbol{\theta})
        \right|
    }_{\mathcal{L}_{\mathrm{stl}}(\boldsymbol{\theta})
    \;(\text{stealth term})}
    \;-\;
    \lambda \cdot
    \underbrace{
        \left|
            \widehat{E}_{\mathrm{cut}}
            (\boldsymbol{\theta})
            -
            \widehat{E}_{\mathrm{full}}
            (\boldsymbol{\theta})
        \right|
    }_{\mathcal{L}_{\mathrm{atk}}(\boldsymbol{\theta})
    \;(\text{attack term})},
    \label{eq:loss}
\end{align}
where $\lambda \geq 0$ governs the trade-off between stealthiness and attack potency.

\noindent\textbf{Stealth term.}
$\mathcal{L}_{\mathrm{stl}}(\boldsymbol{\theta})
= |\widehat{E}_{\mathrm{full}}(\boldsymbol{\theta})|$ drives $\boldsymbol{\theta}$ toward the variational ground state, ensuring the poisoned parameters remain indistinguishable from legitimately trained
parameters under full-circuit validation. No externally known reference energy is required.

\noindent\textbf{Attack term.}
$\mathcal{L}_{\mathrm{atk}}(\boldsymbol{\theta})
= \Delta(\boldsymbol{\theta})$
is the finite-shot evaluation gap from~\eqref{eq:gap}. Maximizing this term drives $\boldsymbol{\theta}$ into regions where the cut-path energy deviates maximally from the full-circuit energy. As established in
\S\ref{sec:finite} and bounded in Theorem~\ref{thm:attack_loss_bound}, this deviation is amplified by $\gamma^2$ due to the reconstruction overhead of the cutting pipeline, while the full-circuit path remains
unaffected.

\noindent\textbf{Why both objectives can be satisfied simultaneously.}
The variance of $\widehat{E}_{\mathrm{full}} (\boldsymbol{\theta})$ scales as $\mathcal{O}(1/N)$, while the variance of
$\widehat{E}_{\mathrm{cut}}(\boldsymbol{\theta})$ scales as $\mathcal{O}(\gamma^2/N)$ where $\gamma$ grows exponentially with $K$. Because
the two terms operate in fundamentally different statistical regimes, minimizing $\mathcal{L}_{\mathrm{stl}}$ does not suppress
$\mathcal{L}_{\mathrm{atk}}$. This scale separation suggests that an optimal $\lambda^*$ exists at which both objectives co-optimize
effectively, as analyzed in \S\ref{sec:lambda}.

\subsection{Optimizing the Backdoor: Role of $\lambda$}
\label{sec:lambda}

The attack weight $\lambda$ governs the balance between 
stealthiness and attack potency in~\eqref{eq:loss}. When 
$\lambda \to 0$, the optimization is dominated by $\mathcal{L}_{\mathrm{stl}}$ and converges to the clean VQA minimum: $\boldsymbol{\theta}$ remains in low-variance regions where the cut-path and full-circuit energies are nearly equal, and the attack is suppressed. When $\lambda \to \infty$, the optimization is dominated by $\mathcal{L}_{\mathrm{atk}}$, pushing $\boldsymbol{\theta}$ far from the clean variational manifold; this degrades $\widehat{E}_{\mathrm{full}}(\boldsymbol{\theta})$ alongside $\widehat{E}_{\mathrm{cut}}(\boldsymbol{\theta})$, causing the poisoned parameters to fail standard validation. At the optimal $\lambda^*$, both objectives co-optimize: 
$\mathcal{L}_{\mathrm{stl}}$ keeps $\widehat{E}_{\mathrm{full}}$ near the variational ground state while $\lambda^* \cdot \mathcal{L}_{\mathrm{atk}}$ drives $\boldsymbol{\theta}$ into regions where $\Delta(\boldsymbol{\theta})$ is large.

\section{Experiments}
\label{sec:exp}

\subsection{Dataset}
\label{sec:dataset}

To evaluate CutBackdoor attack against circuit cutting, we selected benchmarks that require partitioning due to qubit constraints on NISQ hardware. For quantum chemistry applications, we used molecular systems from the PennyLane Molecules dataset~\cite{azad2023pennylane}: $H_3^+$ (6 qubits) and $CH_2$ (14 qubits). Fermionic Hamiltonians were mapped to qubit operators via the Jordan-Wigner transformation~\cite{fradkin1989jordan}, yielding a sum of Pauli string operators acting on the qubit register.. For combinatorial optimization, we constructed QAOA circuits using 8-node MaxCut problem instances~\cite{sawaya2024hamlib}. We also evaluated Variational Quantum Deflation (VQD) on the $H_3^+$ and $H_4$ molecule to assess backdoor robustness. Table~\ref{tab:benchmarks} summarizes circuit characteristics. 

\begin{table}[t]
\centering
\caption{The VQA Benchmarks}
\label{tab:benchmarks}
\begin{tabular}{l c c r}
\toprule
\textbf{Benchmarks} & \textbf{Qubit} & \textbf{1-qubit gate} & \textbf{2-qubit gate} \\
\midrule
VQE - H3+ & 6 & 60 & 9\\
VQE - CH2 & 14 & 140 & 26 \\
VQD - H3+ & 13 & 177 &  66 \\
VQD - H4 & 17 & 275 & 92\\
QAOA -8  & 8 & 72 & 16 \\
\bottomrule
\end{tabular}
\end{table}

%

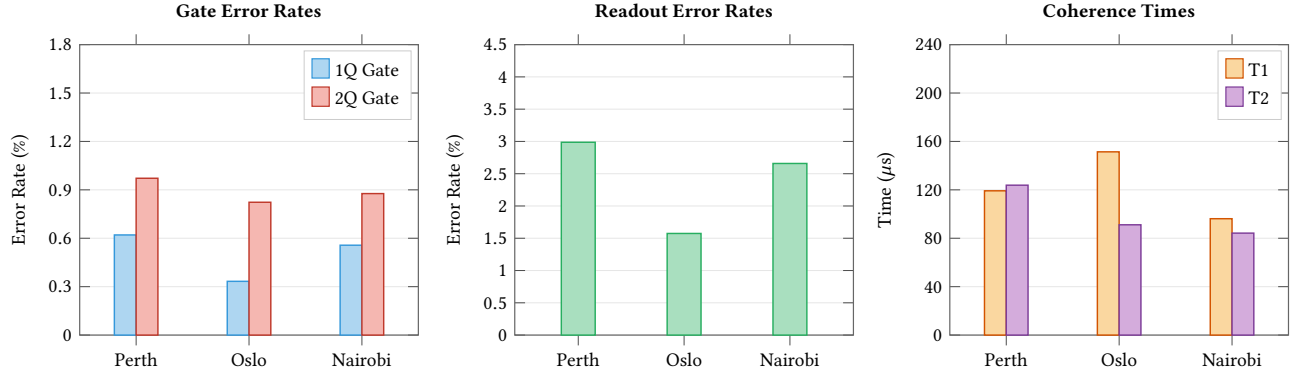
\begin{figure*}[htbp]
  \centering
  \resizebox{\textwidth}{!}{
\definecolor{c1q}{RGB}{174,214,241}   
\definecolor{c2q}{RGB}{245,183,177}   
\definecolor{cro}{RGB}{169,223,191}   
\definecolor{ct1}{RGB}{250,215,160}   
\definecolor{ct2}{RGB}{212,172,220}   

\definecolor{c1qd}{RGB}{52,152,219}
\definecolor{c2qd}{RGB}{192,57,43}
\definecolor{crod}{RGB}{39,174,96}
\definecolor{ct1d}{RGB}{211,84,0}
\definecolor{ct2d}{RGB}{125,60,152}

\begin{tikzpicture}

\begin{axis}[
  name=ax1,
  width=6.5cm, height=5.8cm,
  ybar=0pt,
  bar width=9pt,
  enlarge x limits=0.25,
  ylabel={\small Error Rate (\%)},
  title={\small\bfseries Gate Error Rates},
  title style={yshift=2pt},
  symbolic x coords={Perth,Oslo,Nairobi},
  xtick=data,
  xticklabel style={font=\small, anchor=north},
  ymin=0, ymax=1.8,
  ytick={0,0.3,0.6,0.9,1.2,1.5,1.8},
  yticklabel style={font=\small},
  ylabel style={font=\small, yshift=2pt},
  ymajorgrids=true,
  grid style={draw=gray!20, line width=0.4pt},
  axis background/.style={fill=white},
  legend style={
    at={(0.98,0.98)}, anchor=north east,
    font=\small, draw=gray!40,
    fill=white, inner sep=3pt,
    row sep=1pt,
  },
  legend cell align=left,
  legend image code/.code={
    \draw[#1, line width=0.5pt] (0cm,-0.1cm) rectangle (0.25cm,0.18cm);
  },
  clip=false,
  tick style={color=black!70, line width=0.4pt},
  axis line style={line width=0.5pt, color=black!60},
]

\addplot[ybar, fill=c1q, draw=c1qd, line width=0.6pt]
  coordinates { (Perth,0.6205) (Oslo,0.3331) (Nairobi,0.5569) };
\addlegendentry{1Q Gate}

\addplot[ybar, fill=c2q, draw=c2qd, line width=0.6pt]
  coordinates { (Perth,0.9718) (Oslo,0.8231) (Nairobi,0.8768) };
\addlegendentry{2Q Gate}

\end{axis}

\begin{axis}[
  name=ax2,
  at={(ax1.south east)}, anchor=south west, xshift=1.4cm,
  width=6.5cm, height=5.8cm,
  ybar=0pt,
  bar width=14pt,
  enlarge x limits=0.3,
  ylabel={\small Error Rate (\%)},
  title={\small\bfseries Readout Error Rates},
  title style={yshift=2pt},
  symbolic x coords={Perth,Oslo,Nairobi},
  xtick=data,
  xticklabel style={font=\small, anchor=north},
  ymin=0, ymax=4.5,
  ytick={0,0.5,1.0,1.5,2.0,2.5,3.0,3.5,4.0,4.5},
  yticklabel style={font=\small},
  ylabel style={font=\small, yshift=2pt},
  ymajorgrids=true,
  grid style={draw=gray!20, line width=0.4pt},
  axis background/.style={fill=white},
  clip=false,
  tick style={color=black!70, line width=0.4pt},
  axis line style={line width=0.5pt, color=black!60},
]

\addplot[ybar, fill=cro, draw=crod, line width=0.6pt]
  coordinates { (Perth,2.9871) (Oslo,1.5743) (Nairobi,2.6586) };

\end{axis}

\begin{axis}[
  name=ax3,
  at={(ax2.south east)}, anchor=south west, xshift=1.4cm,
  width=6.5cm, height=5.8cm,
  ybar=0pt,
  bar width=9pt,
  enlarge x limits=0.25,
  ylabel={\small Time ($\mu$s)},
  title={\small\bfseries Coherence Times},
  title style={yshift=2pt},
  symbolic x coords={Perth,Oslo,Nairobi},
  xtick=data,
  xticklabel style={font=\small, anchor=north},
  ymin=0, ymax=240,
  ytick={0,40,80,120,160,200,240},
  yticklabel style={font=\small},
  ylabel style={font=\small, yshift=2pt},
  ymajorgrids=true,
  grid style={draw=gray!20, line width=0.4pt},
  axis background/.style={fill=white},
  legend style={
    at={(0.98,0.98)}, anchor=north east,
    font=\small, draw=gray!40,
    fill=white, inner sep=3pt,
    row sep=1pt,
  },
  legend cell align=left,
  legend image code/.code={
    \draw[#1, line width=0.5pt] (0cm,-0.1cm) rectangle (0.25cm,0.18cm);
  },
  clip=false,
  tick style={color=black!70, line width=0.4pt},
  axis line style={line width=0.5pt, color=black!60},
]

\addplot[ybar, fill=ct1, draw=ct1d, line width=0.6pt]
  coordinates { (Perth,119.20) (Oslo,151.40) (Nairobi,96.15) };
\addlegendentry{T1}

\addplot[ybar, fill=ct2, draw=ct2d, line width=0.6pt]
  coordinates { (Perth,123.85) (Oslo,91.12) (Nairobi,84.24) };
\addlegendentry{T2}

\end{axis}

\end{tikzpicture}}
  \caption{Noise characterization of the three 7-qubit
  IBMQ. Gate error rates (left), readout error rates
  (center), and coherence times (right) for
  \texttt{IBMQ\_Perth}, \texttt{IBMQ\_Oslo}, and
  \texttt{IBMQ\_Nairobi}.}
  \label{fig:device_noise}
\end{figure*}

\subsection{Circuit Cutting Configuration}
\label{subsec:cutting-constraints}

All experiments use CutQC~\cite{tang2021cutqc} as the circuit  cutting framework, which invokes the Gurobi optimizer~\cite{gurobi2026}  to solve a Mixed-Integer Program (MIP) that identifies optimal cut  locations minimizing classical reconstruction cost, replicating the  automatic cutting procedure a real practitioner would invoke. Cutting  behavior is governed by the cutter\_constraints dictionary  (Table~\ref{tab:cutter_constraints}). The max subcircuit width  parameter enforces the qubit capacity of the target backend; max cuts caps the total wire-cut budget to bound reconstruction overhead, which scales as $\mathcal{O}(4^k)$ in the number of cuts~$k$. The num subcircuits parameter is supplied as an ordered list, allowing Gurobi to select the smallest feasible partition automatically. The max subcircuit cuts parameter limits cut edges per fragment to keep the reconstruction branch count tractable, and subcircuit size imbalance grants the solver flexibility for molecular circuits whose qubit counts do not partition evenly.

\subsection{Quantum Devices and Compilations}
\label{sec:devices}

All experiments are conducted on IBM quantum processors via Qiskit, with circuits transpiled at \texttt{optimization\_level=0} or \texttt{1} using a trivial initial layout to preserve the circuit structure intended by the attacker. To evaluate $E_{\mathrm{stl}}$, the full circuit is executed on \texttt{IBMQ\_Kolkata} (27q), which provides sufficient capacity to run the target ansatz without cutting and mirrors the victim's validation environment. To evaluate $E_{\mathrm{abs}}$, subcircuits produced by the CutQC pipeline are executed on three 7q processors: \texttt{IBMQ\_Perth}, \texttt{IBMQ\_Oslo}, and \texttt{IBMQ\_Nairobi}. As shown in Figure~\ref{fig:device_noise}, these devices exhibit meaningfully different noise profiles: \texttt{IBMQ\_Perth} has the highest two-qubit gate error rate 
(${\sim}1.0\%$) and readout error rate (${\sim}3.0\%$), while \texttt{IBMQ\_Oslo} is the lowest across both metrics and exhibits the longest $T_1$ (${\sim}150\,\mu$s). For generality analyses on smaller hardware, we additionally run selected experiments on three 5q processors: \texttt{IBMQ\_Manila}, \texttt{IBMQ\_Lima}, and 
\texttt{IBMQ\_Quito}.


\begin{table}[htbp]
\centering
\caption{Cutting constraints for n qubits circuits}
\label{tab:cutter_constraints}
\small
\begin{tabular}{lll}
\toprule
\textbf{Parameter} & \textbf{Role} & \textbf{Value} \\
\midrule
max\_subcircuit\_width      & Max qubits per subcircuit     & 6 \\
max\_cuts                   & Total wire-cut budget         & 10           \\
num\_subcircuits            & Partition counts to try       & $[2,\,3,\,4]$     \\
max\_subcircuit\_cuts       & Cut edges per subcircuit      & 6                 \\
subcircuit\_size\_imbalance & MIP balance relaxation        & 2               \\
\bottomrule
\end{tabular}
\end{table}

\subsection{Circuit Benchmarks and Training}
\label{sec:benchmarks}

Table~\ref{tab:benchmarks} summarizes the quantum circuits for  the five representative VQAs used in our evaluation, spanning  circuit sizes from 6 to 17 qubits with diverse ansatz  architectures. For VQE tasks, we adopt the ansatz proposed  in~\cite{tilly2022variational}; for QAOA, we follow the circuit  design in~\cite{wu2021towards}; and for VQD, we employ the  framework described in~\cite{higgott2019variational}. These  architectural differences yield single-qubit gate counts ranging  from 60 to 275 and two-qubit gate counts ranging from 9 to 92.  All VQAs were trained using Qiskit with the Adam optimizer, a  learning rate of $5\times10^{-3}$, momentum parameters  $\beta_1 = 0.9$ and $\beta_2 = 0.999$, and a perturbation coefficient of $c = 0.06$ for Simultaneous Perturbation Stochastic Approximation (SPSA) gradient estimation, which approximates gradients using only two circuit evaluations per step regardless of parameter count.

\subsection{Noise Mitigation}
\label{sec:zne}

To evaluate the robustness of CutBackdoor under error mitigation, we apply Zero Noise Extrapolation (ZNE)~\cite{temme2017error} using the Mitiq framework~\cite{larose2022mitiq}. ZNE estimates the noise-free expectation value by executing the circuit at multiple artificially scaled noise levels and extrapolating back to the zero-noise limit. For each circuit, we generate noisy variants at scaling factors $T \in \{1, 2, 3, 4, 5, 6\}$ and fit a second-degree polynomial over the sampled expectation values to obtain the extrapolated result. This setting is applied to both the full-circuit and cut-based execution paths, allowing us to assess whether the backdoor remains effective even when the victim employs standard error mitigation alongside circuit cutting.

\subsection{Evaluation Metrics}
\label{sec:metrics}

To quantify attack effectiveness and stealthiness, we define two evaluation metrics. For a given parameter set $\boldsymbol{\theta}$, we define:
\begin{align}
  E_{\mathrm{abs}} &=
    \bigl|\hat{E}_{\mathrm{cut}}(\boldsymbol{\theta})
          - E_{\mathrm{ideal}}(\boldsymbol{\theta})\bigr|,
    \label{eq:metric_atk} \\
  E_{\mathrm{stl}} &=
    \bigl|\hat{E}_{\mathrm{full}}(\boldsymbol{\theta})
          - E_{\mathrm{ideal}}(\boldsymbol{\theta})\bigr|,
    \label{eq:metric_stl}
\end{align}
where $E_{\mathrm{ideal}}(\boldsymbol{\theta})$ is the noiseless statevector expectation value of the full uncut circuit under $\boldsymbol{\theta}$. $E_{\mathrm{abs}}$ measures how far the reconstructed cut-path energy deviates from this ideal value, while $E_{\mathrm{stl}}$ measures how far the full-circuit energy deviates from the same baseline. A successful attack requires $E_{\mathrm{abs}}$ to be substantially elevated under backdoor parameters while $E_{\mathrm{stl}}$ remains small.

\begin{table*}[t]
\centering
\setlength{\tabcolsep}{6pt}
\renewcommand{\arraystretch}{0.95}
\caption{%
  CutBackdoor effectiveness across VQE, VQD, and QAOA benchmarks.
  $E_{\mathrm{stl}}$ (no-cut, \texttt{IBMQ\_Kolkata}) confirms stealthiness;
  $E_{\mathrm{abs}}$ (with-cut) measures attack potency. Parentheses give the
  CutBackdoor/clean ratio and \revA{the absolute gap
  $\Delta = E_{\mathrm{abs}}^{\mathrm{BD}} - E_{\mathrm{abs}}^{\mathrm{clean}}$.}
}
\label{tab:main_results}
\begin{footnotesize}
\begin{tabular}{ll c cccc}
\toprule
\multirow{2}{*}{\textbf{Benchmark}}
  & \multirow{2}{*}{\textbf{Schemes}}
  & \textbf{No-cut}
  & \multicolumn{4}{c}{\textbf{With-cut $E_{\mathrm{abs}}$}} \\
\cmidrule(lr){3-3}\cmidrule(lr){4-7}
  &
  & $E_{\mathrm{stl}}$
  & Kolkata & Perth & Oslo & Nairobi \\
\midrule
\multirow{2}{*}{VQE -- H$_3^+$}
  & Clean        & 0.076 & 0.098 & 0.166 & 0.122 & 0.193 \\
  & CutBackdoor  & 0.108\,{\scriptsize(1.4$\times$)}
                 & 0.235\,{\scriptsize(2.4$\times$,\,$\Delta{+}0.137$)}
                 & 0.313\,{\scriptsize(1.9$\times$,\,$\Delta{+}0.147$)}
                 & 0.262\,{\scriptsize(2.2$\times$,\,$\Delta{+}0.140$)}
                 & 0.305\,{\scriptsize(1.6$\times$,\,$\Delta{+}0.112$)} \\
\midrule
\multirow{2}{*}{VQE -- CH$_2$}
  & Clean        & 4.183 & 0.762 & 0.743 & 1.052 & 1.776 \\
  & CutBackdoor  & 3.792\,{\scriptsize(0.9$\times$)}
                 & 1.138\,{\scriptsize(1.5$\times$,\,$\Delta{+}0.376$)}
                 & 1.454\,{\scriptsize(2.0$\times$,\,$\Delta{+}0.711$)}
                 & 3.020\,{\scriptsize(2.9$\times$,\,$\Delta{+}1.968$)}
                 & 3.152\,{\scriptsize(1.8$\times$,\,$\Delta{+}1.376$)} \\
\midrule
\multirow{2}{*}{VQD -- H$_3^+$}
  & Clean        & 0.099 & 0.151 & 0.198 & 0.152 & 0.230 \\
  & CutBackdoor  & 0.104\,{\scriptsize(1.1$\times$)}
                 & 0.315\,{\scriptsize(2.1$\times$,\,$\Delta{+}0.164$)}
                 & 0.385\,{\scriptsize(2.0$\times$,\,$\Delta{+}0.187$)}
                 & 0.363\,{\scriptsize(2.4$\times$,\,$\Delta{+}0.211$)}
                 & 0.316\,{\scriptsize(1.4$\times$,\,$\Delta{+}0.086$)} \\
\midrule
\multirow{2}{*}{VQD -- H$_4$}
  & Clean        & 0.166 & 0.262 & 0.470 & 0.300 & 0.454 \\
  & CutBackdoor  & 0.151\,{\scriptsize(0.9$\times$)}
                 & 0.342\,{\scriptsize(1.3$\times$,\,$\Delta{+}0.080$)}
                 & 0.680\,{\scriptsize(1.5$\times$,\,$\Delta{+}0.210$)}
                 & 0.449\,{\scriptsize(1.5$\times$,\,$\Delta{+}0.149$)}
                 & 0.649\,{\scriptsize(1.4$\times$,\,$\Delta{+}0.195$)} \\
\midrule
\multirow{2}{*}{QAOA -- 8}
  & Clean        & 0.356 & 0.269 & 0.947 & 0.574 & 0.779 \\
  & CutBackdoor  & 0.380\,{\scriptsize(1.1$\times$)}
                 & 0.459\,{\scriptsize(1.7$\times$,\,$\Delta{+}0.190$)}
                 & 0.990\,{\scriptsize(1.0$\times$,\,$\Delta{+}0.043$)}
                 & 0.564\,{\scriptsize(1.0$\times$,\,$\Delta{-}0.010$)}
                 & 1.042\,{\scriptsize(1.3$\times$,\,$\Delta{+}0.263$)} \\
\bottomrule
\end{tabular}
\end{footnotesize}
\end{table*}

\section{Experimental Results}
\label{sec:results}

\subsection{Attack Efficiency}
\label{sec:efficiency}

We evaluate CutBackdoor across five VQA benchmarks executed on four IBM quantum processors under clean and backdoor parameter configurations. Table~\ref{tab:main_results} reports $E_{\mathrm{stl}}$ and $E_{\mathrm{abs}}$ in Hartree (Ha) for each benchmark and processor.

\textbf{VQE.}
For the primary VQE benchmark on the trihydrogen cation (H$_3^+$), CutBackdoor produces consistently elevated cut-path errors across all four processors. The no-cut stealthiness error increases modestly from 0.076 to 0.108, remaining within a range that would not raise concern during full-circuit validation. On the cut-based path, $E_{\mathrm{abs}}$ rises from 0.235 on \texttt{IBMQ\_Kolkata} to 0.313 on \texttt{IBMQ\_Perth} against clean baselines of 0.098 and 0.166, yielding amplification ratios of $1.6\times$ to $2.4\times$. The largest amplification occurs on the lowest-noise device, indicating that the attack signal emerges most cleanly when the cutting pipeline is the dominant error source rather than hardware noise. The larger CH$_2$ benchmark involves greater circuit depth and exhibits elevated clean baselines from accumulated noise, yet CutBackdoor further amplifies cut-path errors by $1.5\times$ on \texttt{IBMQ\_Kolkata} to $2.9\times$ on \texttt{IBMQ\_Oslo}, showing the attack surface persists at scale. Notably, $E_{\mathrm{stl}}$ under CutBackdoor (3.792) falls slightly below the clean baseline (4.183), confirming that the dual-objective loss anchors the full-circuit energy near the variational ground state even as it drives the cut-path energy toward the wrong target.

\textbf{VQD.}
On excited-state benchmarks, CutBackdoor remains effective despite the
additional orthogonality constraints VQD imposes on the ansatz. For VQD--H$_3^+$, amplification ranges from $1.4\times$ to $2.4\times$ across processors, with the largest absolute gain of 0.211 on \texttt{IBMQ\_Oslo}, while $E_{\mathrm{stl}}$ increases by only 0.005, the cleanest separation between attack and stealth in the table. For the VQD--H$_4$ ansatz, amplification is more modest at $1.3\times$--$1.5\times$, with absolute error increases from 0.080 on \texttt{IBMQ\_Kolkata} to 0.210 on \texttt{IBMQ\_Perth}; the narrower range reflects the higher clean baseline leaving less relative headroom, but CutBackdoor still elevates the cut-path error above the clean baseline on every device. In both cases the no-cut error stays near the clean baseline, confirming that the attack extends cleanly to excited-state tasks without sacrificing stealth.

\textbf{QAOA.}
For the MaxCut QAOA benchmark, CutBackdoor produces amplification ratios of $1.0\times$ to $1.7\times$ across backends, with the largest absolute increase of 0.263 on \texttt{IBMQ\_Nairobi}. The effect is more modest than on VQE and VQD, where ratios reach $2.9\times$. This contrast is consistent with the underlying mechanism: the parameter regions the attack exploits are accessible primarily when the observable couples non-trivially to the cut wire, which is typical for molecular Hamiltonians but less favorable for diagonal cost Hamiltonians like MaxCut. QAOA therefore bounds the attack's applicability rather than serving as a primary target.

Across the VQE and VQD benchmarks, CutBackdoor consistently elevates cut-path errors while maintaining full-circuit stealthiness, with amplification of $1.3\times$ to $2.9\times$ depending on processor and benchmark, \revA{while the diagonal-cost QAOA benchmark marks the attack's structural boundary.}


\subsection{Empirical Validation of the Finite-Shot Bound}
\label{sec:shot_validation}
Theorem~\ref{thm:attack_loss_bound} makes two predictions: the realized gap contracts as $\mathcal{O}(1/\sqrt{N})$ in the shot budget, and its ceiling widens with the cut count $K$ through the overhead $\gamma$. We validate both, sweeping $N$ at fixed $K$ (Fig.~\ref{fig:shots_eval}) and $K$ at fixed $N$ (Table~\ref{tab:kablation}).

\noindent\textbf{Shot budget.}
Fixing a CutBackdoor configuration $\boldsymbol{\theta}^*$ on VQE--H$_3^+$, we record 20 trials of $\Delta(\boldsymbol{\theta}^*) = |\hat{E}_{\mathrm{full}} - \hat{E}_{\mathrm{cut}}|$ per shot level from $1\mathrm{k}$ to $64\mathrm{k}$ (Fig.~\ref{fig:shots_eval}). The $\bar{\Delta} \pm \sigma$ band contracts as $N$ grows, matching the $\mathcal{O}(1/\sqrt{N})$ shrinkage, while the mean $\bar{\Delta}$ stays near $175\,$mHa at every $N$. This persistence is the signature of dual-objective training: the loss in~\eqref{eq:loss} steers $\boldsymbol{\theta}^*$ into high-variance regions of the cut-path estimator, whereas clean parameters sit in quieter regions and yield a much smaller $\bar{\Delta}$. The bound holds in both cases; what differs is whether the realized gap sits near its ceiling or far below it.

\noindent\textbf{Cut count.}

\revA{Holding $\boldsymbol{\theta}^*$ and the circuit fixed, we re-run CutQC under three partitions at a fixed shot budget (Table~\ref{tab:kablation}). The gap grows monotonically with $K$ ($178 \to 232 \to 242\,$mHa), confirming that more cuts genuinely widen the attack surface. Yet this growth is far slower than the bound: as $K$ rises from $1$ to $5$ the overhead $\gamma$ grows from roughly $16$ to $864$, expanding the $(\gamma{+}1)$ prefactor of Theorem~\ref{thm:attack_loss_bound} by over an order of magnitude, while $\bar{\Delta}$ rises only $1.36\times$. The increments diminish ($+54$ vs.\ $+10\,$mHa), consistent with the $M = 4^K$ branch count making training progressively harder. The realized gap is thus a nontrivial but sub-maximal fraction of the one-sided bound: Theorem~\ref{thm:attack_loss_bound} predicts the direction of the $K$-dependence and caps the worst case without claiming the attack saturates it.}

\begin{figure}[t]
  \centering
  \includegraphics[width=\linewidth]{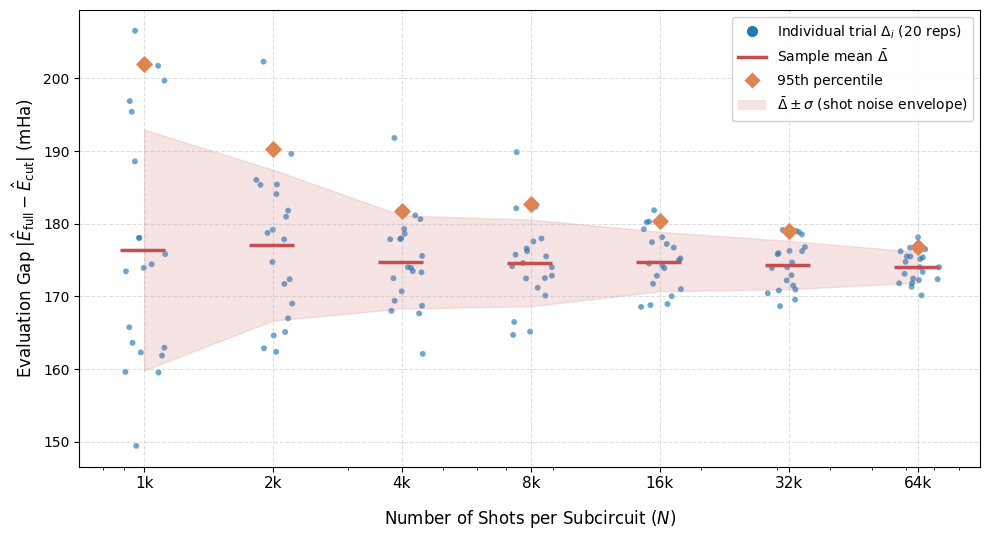}
  \Description{Plot showing finite-shot gap Delta(theta*) under backdoor parameters on IBMQ_Kolkata, with error bands and mean values as shot level increases.}
  \caption{Finite-shot gap $\Delta(\boldsymbol{\theta}^*)$ under backdoor
parameters on IBMQ\_Kolkata (20 trials per shot level). The
$\bar{\Delta} \pm \sigma$ band contracts as $\mathcal{O}(1/\sqrt{N})$
(Theorem~\ref{thm:attack_loss_bound}), while $\bar{\Delta}$ stays elevated.}
  \label{fig:shots_eval}
\end{figure}

\begin{table}[t]
  \centering
 \caption{\revA{Wire-cut ablation on VQE--H$_3^+$ at fixed $\boldsymbol{\theta}^*$ and shot budget. $\bar{\Delta}$ grows with $K$ but stays far below the $\mathcal{O}(\gamma)$ ceiling of Theorem~\ref{thm:attack_loss_bound}.}}
  \label{tab:kablation}
  \begin{tabular}{cccc}
    \toprule
    Subcircuits & $K$ & $M = 4^K$ & mean $\bar{\Delta}$ (mHa) \\
    \midrule
    2 & 1 & 4    & 177.8 \\
    3 & 3 & 64   & 231.9 \\
    4 & 5 & 1024 & 242.0 \\
    \bottomrule
  \end{tabular}
\end{table}

\subsection{Stealthiness}
\label{sec:stealth}
\subsubsection{Effect of Zero-Noise Extrapolation}
To assess whether CutBackdoor remains effective when the victim applies error mitigation alongside circuit cutting, we evaluate all benchmarks under three ZNE fitting methods: Linear, Polynomial, and Exponential extrapolation. Table~\ref{tab:zne_results} reports $E_{\mathrm{stl}}$ and $E_{\mathrm{abs}}$ for clean and CutBackdoor under each fitting method. Across all three ZNE variants, CutBackdoor consistently produces elevated cut-path errors relative to clean parameters on the VQE and VQD benchmarks. Under Linear extrapolation, CutBackdoor achieves amplification of $2.3\times$ on VQE--H$_3^+$ and $2.1\times$ on VQD--H$_3^+$, while stealthiness errors increase only modestly, confirming that ZNE does not neutralize the adversarial signal. Under Polynomial fitting, the most pronounced amplification is observed on VQD--H$_4$, where the cut-path error reaches 0.495 against a clean baseline of 0.031, a $15.8\times$ increase, indicating that Polynomial extrapolation can inadvertently amplify the adversarial signal when the noise-scaling curve is distorted by the attack. Similarly, under Exponential fitting, VQD--H$_3^+$ exhibits a $9.3\times$ amplification, with $E_{\mathrm{abs}}$ increasing from 0.034 to 0.318, the largest relative amplification observed across all ZNE variants. In contrast, the QAOA benchmark shows limited amplification across all three fitting methods, with $E_{\mathrm{abs}}$ ratios ranging from $0.7\times$ to $1.2\times$, substantially lower than on VQE and VQD benchmarks. Overall, ZNE does not constitute a complete defense against CutBackdoor on most benchmarks, though it substantially weakens the attack on VQE--CH$_2$ and QAOA.

\subsubsection{Effect of Compilation Settings}
\label{sec:compilation}
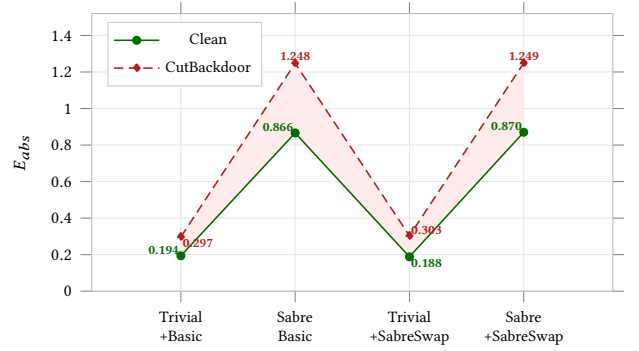
\begin{figure}[t]
  \centering
\definecolor{darkgreen}{RGB}{0,110,0}
\definecolor{attackred}{RGB}{180,30,30}

\begin{tikzpicture}
\begin{axis}[
    width=\columnwidth,          
    height=0.62\columnwidth,     
    ylabel={$E_{abs}$ },
    ylabel style={font=\scriptsize, yshift=2pt},
    ymin=0.0, ymax=1.52,
    xmin=0.6, xmax=4.4,
    xtick={1,2,3,4},
    xticklabels={
        {\scriptsize Trivial\\\scriptsize +Basic},
        {\scriptsize Sabre\\\scriptsize Basic},
        {\scriptsize Trivial\\\scriptsize +SabreSwap},
        {\scriptsize Sabre\\\scriptsize +SabreSwap}
    },
    xticklabel style={align=center, font=\scriptsize},
    yticklabel style={font=\scriptsize},
    ytick={0.0,0.2,0.4,0.6,0.8,1.0,1.2,1.4},
    tick label style={font=\scriptsize},
    ymajorgrids=true,
    xmajorgrids=true,
    grid style={gray!20, very thin},
    tick align=outside,
    axis line style={gray!55, thin},
    legend style={
        at={(0.03,0.97)},
        anchor=north west,
        font=\scriptsize,
        draw=gray!40,
        fill=white,
        inner sep=2pt,
        row sep=1pt,
    },
    clip=false,
    enlarge x limits=0.1,
]

\addplot [fill=red!8, draw=none, forget plot] coordinates {
    (1, 0.193781)(2, 0.866176)(3, 0.188399)(4, 0.869963)
    (4, 1.249095)(3, 0.303265)(2, 1.247843)(1, 0.297419)
} -- cycle;

\addplot [
    color=darkgreen, mark=*, line width=0.6pt,
    mark options={fill=darkgreen, draw=darkgreen, scale=0.7},
] coordinates {
    (1, 0.193781)(2, 0.866176)(3, 0.188399)(4, 0.869963)
};
\addlegendentry{Clean}

\addplot [
    color=attackred, dashed, line width=0.6pt,
    dash pattern=on 3.5pt off 2pt,
    mark=diamond*, mark options={fill=attackred, draw=attackred, scale=0.8},
] coordinates {
    (1, 0.297419)(2, 1.247843)(3, 0.303265)(4, 1.249095)
};
\addlegendentry{CutBackdoor}

\node[font=\tiny\bfseries, text=darkgreen, above left,  inner sep=0.5pt] at (axis cs:1, 0.193781) {0.194};
\node[font=\tiny\bfseries, text=darkgreen, above left,  inner sep=0.5pt] at (axis cs:2, 0.866176) {0.866};
\node[font=\tiny\bfseries, text=darkgreen, below right, inner sep=0.5pt] at (axis cs:3, 0.188399) {0.188};
\node[font=\tiny\bfseries, text=darkgreen, above left,  inner sep=0.5pt] at (axis cs:4, 0.869963) {0.870};

\node[font=\tiny\bfseries, text=attackred, below right, inner sep=0.5pt] at (axis cs:1, 0.297419) {0.297};
\node[font=\tiny\bfseries, text=attackred, above,       inner sep=1pt]   at (axis cs:2, 1.247843) {1.248};
\node[font=\tiny\bfseries, text=attackred, above right, inner sep=0.5pt] at (axis cs:3, 0.303265) {0.303};
\node[font=\tiny\bfseries, text=attackred, above,       inner sep=1pt]   at (axis cs:4, 1.249095) {1.249};





\end{axis}
\end{tikzpicture}
  \caption{Effect of CutBackdoor across different compilation settings.}
  \label{fig:compilation}
\end{figure}

To evaluate CutBackdoor under real-world deployment conditions, we test the VQE--H$_3^+$ ansatz on \texttt{IBMQ\_Manila} across four compilation configurations, combining two layout passes (Trivial, Sabre) with two routing passes (Basic, SabreSwap). Figure~\ref{fig:compilation} reveals two findings. First, compilation mismatch alone degrades energy accuracy for both clean and CutBackdoor parameters, independent of any attack: under Sabre layout, clean parameters incur $E_{\mathrm{abs}}\!\approx\!0.87$ versus $\approx\!0.19$ under trivial layout, driven by the additional SWAP overhead from aggressive routing. Second, the attack gap between clean and CutBackdoor persists across all four configurations. Under trivial layout, which matches the attacker's training environment, CutBackdoor produces $E_{\mathrm{abs}}\!\approx\!0.297$ against a clean baseline of $\approx\!0.194$, confirming the adversarial effect under matched compilation.

The decisive test is whether the effect is a property of trivial layout specifically or of matched compilation in general. \revA{We therefore train and deploy entirely under \texttt{SabreLayout}\,+\,\texttt{SabreSwap} (optimization level~3): here CutBackdoor reaches $E_{\mathrm{abs}}=0.118$ against a clean baseline of $0.019$ ($\Delta E_{\mathrm{abs}}=0.099$), while stealth stays flat ($E_{\mathrm{stl}}=0.041$ versus $0.040$). The attack
therefore persists under a fully Sabre-compiled pipeline, confirming that CutBackdoor is \emph{compilation-relative}: its effectiveness requires the attacker's training to match the victim's deployment rather than any specific layout, and it is not tied to trivial routing. The only requirement is that attacker and victim share a compilation configuration, which the standard CutQC workflow satisfies by default.}

\begin{table*}[t]
\centering
\setlength{\tabcolsep}{4pt}
\renewcommand{\arraystretch}{1}
\caption{%
  CutBackdoor robustness under Zero-Noise Extrapolation (ZNE).
  For each ZNE fitting method, $E_{\mathrm{stl}}$ (no-cut) and
  $E_{\mathrm{abs}}$ (with-cut) are reported in Ha.
  A successful attack requires large $E_{\mathrm{abs}}$ with small
  $E_{\mathrm{stl}}$ simultaneously.
  Parentheses denote the CutBackdoor\,/\,clean ratio;
  \revA{\textit{$\Delta$\,Error} gives the absolute deviation from Clean.}
}
\label{tab:zne_results}
\begin{footnotesize}
\begin{tabular}{l l  cc  cc  cc  cc  cc}
\toprule
& &
\multicolumn{2}{c}{\textbf{VQE -- H$_3^+$}} &
\multicolumn{2}{c}{\textbf{VQE -- CH$_2$}} &
\multicolumn{2}{c}{\textbf{VQD -- H$_3^+$}} &
\multicolumn{2}{c}{\textbf{VQD -- H$_4$}} &
\multicolumn{2}{c}{\textbf{QAOA -- 8}} \\
\cmidrule(lr){3-4}\cmidrule(lr){5-6}
\cmidrule(lr){7-8}\cmidrule(lr){9-10}\cmidrule(lr){11-12}
\textbf{ZNE Method} & \textbf{Schemes}
  & $E_{\mathrm{stl}}$ & $E_{\mathrm{abs}}$
  & $E_{\mathrm{stl}}$ & $E_{\mathrm{abs}}$
  & $E_{\mathrm{stl}}$ & $E_{\mathrm{abs}}$
  & $E_{\mathrm{stl}}$ & $E_{\mathrm{abs}}$
  & $E_{\mathrm{stl}}$ & $E_{\mathrm{abs}}$ \\
\midrule
\multirow{3}{*}{Linear}
  & Clean
    & 0.055 & 0.080 & 2.860 & 1.780 & 0.181 & 0.164 & 0.099 & 0.253 & 0.383 & 0.751 \\
  & CutBackdoor
    & 0.080\,{\scriptsize(1.5$\times$)} & 0.186\,{\scriptsize(2.3$\times$)}
    & 2.981\,{\scriptsize(1.0$\times$)} & 1.788\,{\scriptsize(1.0$\times$)}
    & 0.207\,{\scriptsize(1.1$\times$)} & 0.342\,{\scriptsize(2.1$\times$)}
    & 0.145\,{\scriptsize(1.5$\times$)} & 0.444\,{\scriptsize(1.8$\times$)}
    & 0.307\,{\scriptsize(0.8$\times$)} & 0.556\,{\scriptsize(0.7$\times$)} \\
\cmidrule(lr){2-12}
  & {\scriptsize\textit{$\Delta$ Error}}
    & {\scriptsize\textit{+0.025}} & {\scriptsize\textit{+0.106}}
    & {\scriptsize\textit{+0.121}} & {\scriptsize\textit{+0.008}}
    & {\scriptsize\textit{+0.026}} & {\scriptsize\textit{+0.178}}
    & {\scriptsize\textit{+0.046}} & {\scriptsize\textit{+0.191}}
    & {\scriptsize\textit{$-$0.076}} & {\scriptsize\textit{$-$0.195}} \\
\midrule
\multirow{3}{*}{Poly}
  & Clean
    & 0.038 & 0.086 & 2.615 & 1.812 & 0.149 & 0.154 & 0.111 & 0.031 & 0.361 & 0.634 \\
  & CutBackdoor
    & 0.086\,{\scriptsize(2.3$\times$)} & 0.155\,{\scriptsize(1.8$\times$)}
    & 3.189\,{\scriptsize(1.2$\times$)} & 1.943\,{\scriptsize(1.1$\times$)}
    & 0.169\,{\scriptsize(1.1$\times$)} & 0.258\,{\scriptsize(1.7$\times$)}
    & 0.086\,{\scriptsize(0.8$\times$)} & 0.495\,{\scriptsize(15.8$\times$)}
    & 0.275\,{\scriptsize(0.8$\times$)} & 0.792\,{\scriptsize(1.2$\times$)} \\
\cmidrule(lr){2-12}
  & {\scriptsize\textit{$\Delta$ Error}}
    & {\scriptsize\textit{+0.048}} & {\scriptsize\textit{+0.069}}
    & {\scriptsize\textit{+0.574}} & {\scriptsize\textit{+0.131}}
    & {\scriptsize\textit{+0.020}} & {\scriptsize\textit{+0.104}}
    & {\scriptsize\textit{$-$0.025}} & {\scriptsize\textit{+0.464}}
    & {\scriptsize\textit{$-$0.086}} & {\scriptsize\textit{+0.158}} \\
\midrule
\multirow{3}{*}{Exponential}
  & Clean
    & 0.055 & 0.058 & 3.047 & 5.532 & 0.150 & 0.034 & 0.113 & 0.354 & 0.452 & 2.156 \\
  & CutBackdoor
    & 0.058\,{\scriptsize(1.0$\times$)} & 0.213\,{\scriptsize(3.7$\times$)}
    & 3.190\,{\scriptsize(1.0$\times$)} & 5.605\,{\scriptsize(1.0$\times$)}
    & 0.174\,{\scriptsize(1.2$\times$)} & 0.318\,{\scriptsize(9.3$\times$)}
    & 0.138\,{\scriptsize(1.2$\times$)} & 0.619\,{\scriptsize(1.7$\times$)}
    & 0.162\,{\scriptsize(0.4$\times$)} & 2.689\,{\scriptsize(1.2$\times$)} \\
\cmidrule(lr){2-12}
  & {\scriptsize\textit{$\Delta$ Error}}
    & {\scriptsize\textit{+0.003}} & {\scriptsize\textit{+0.155}}
    & {\scriptsize\textit{+0.143}} & {\scriptsize\textit{+0.073}}
    & {\scriptsize\textit{+0.024}} & {\scriptsize\textit{+0.284}}
    & {\scriptsize\textit{+0.025}} & {\scriptsize\textit{+0.265}}
    & {\scriptsize\textit{$-$0.290}} & {\scriptsize\textit{+0.533}} \\
\bottomrule
\end{tabular}
\end{footnotesize}
\end{table*}




\subsection{Generality}
\label{sec:generality}
\begin{figure}[t]
  \centering
  \pgfplotsset{compat=1.18}

\definecolor{CleanColor}{RGB}{0, 130, 100}
\definecolor{BackColor}{RGB}{196, 72, 28}
\definecolor{AxisGray}{RGB}{90, 90, 90}
\definecolor{GridGray}{RGB}{215, 215, 215}
\definecolor{LabelGreen}{RGB}{0, 115, 85}
\definecolor{LabelRed}{RGB}{185, 55, 15}


\begin{tikzpicture}
\begin{axis}[
  width       = \columnwidth,
  height      = 0.60\columnwidth,
  xmin = 0.6, xmax = 3.4,
  xtick       = {1, 2, 3},
  xticklabels = {
    {CutQC\\[-1.5pt]\scriptsize(3+3)},
    {Manual\\[-1.5pt]\scriptsize(2+4)},
    {Manual\\[-1.5pt]\scriptsize(4+2)}
  },
  xticklabel style = {align=center, font=\footnotesize, text depth=1.8ex},
  ymin = 0.11, ymax = 0.44,
  ytick = {0.15,0.20,0.25,0.30,0.35,0.40},
  yticklabel style = {font=\footnotesize},
  ylabel       = {$E_{abs}$ },
  ylabel style = {font=\footnotesize},
  grid             = both,
  grid style       = {black!11, very thin},
  minor tick num   = 1,
  minor grid style = {black!6, ultra thin},
  axis line style  = {black!55, thin},
  tick style       = {black!55, thin},
  tick align       = outside,
  legend pos        = north east,
  legend cell align = left,
  legend style = {
    draw        = black!30,
    line width  = 0.4pt,
    fill        = white,
    fill opacity= 0.93,
    text opacity= 1,
    font        = \scriptsize,
    row sep     = -2pt,
    inner xsep  = 4pt,
    inner ysep  = 2pt,
  },
  clip = false,
]

\addplot [fill=red!14, draw=none, forget plot] coordinates {
  (1,0.187851)(2,0.226676)(3,0.158708)
  (3,0.296250)(2,0.382652)(1,0.293249)(1,0.187851)
};

\addplot [
  color=black!65!green, line width=1.0pt,
  mark=*, mark size=2.5pt,
  mark options={fill=black!65!green, draw=black!65!green, line width=0.5pt},
] coordinates {(1,0.187851)(2,0.226676)(3,0.158708)};
\addlegendentry{Clean}

\addplot [
  color=red!80!black, line width=1.0pt,
  dashed, dash pattern=on 4.5pt off 2.2pt,
  mark=diamond*, mark size=3.2pt,
  mark options={fill=red!80!black, draw=red!80!black, line width=0.5pt},
] coordinates {(1,0.293249)(2,0.382652)(3,0.296250)};
\addlegendentry{CutBackdoor}

\node [font=\scriptsize\bfseries, color=black!70!green,
       anchor=north, yshift=-3pt] at (axis cs:1,0.187851) {0.1879};
\node [font=\scriptsize\bfseries, color=black!70!green,
       anchor=south, yshift=3pt]  at (axis cs:2,0.226676) {0.2267};
\node [font=\scriptsize\bfseries, color=black!70!green,
       anchor=north, yshift=-3pt] at (axis cs:3,0.158708) {0.1587};

\node [font=\scriptsize\bfseries, color=red!80!black,
       anchor=south, yshift=3pt]  at (axis cs:1,0.293249) {0.2932};
\node [font=\scriptsize\bfseries, color=red!80!black,
       anchor=south, yshift=3pt]  at (axis cs:2,0.382652) {0.3827};
\node [font=\scriptsize\bfseries, color=red!80!black,
       anchor=south, yshift=3pt]  at (axis cs:3,0.296250) {0.2963};

\end{axis}
\end{tikzpicture}   
  \caption{Effect of CutBackdoor with different Cut locations}
  \label{fig:cut_strategy}
\end{figure}
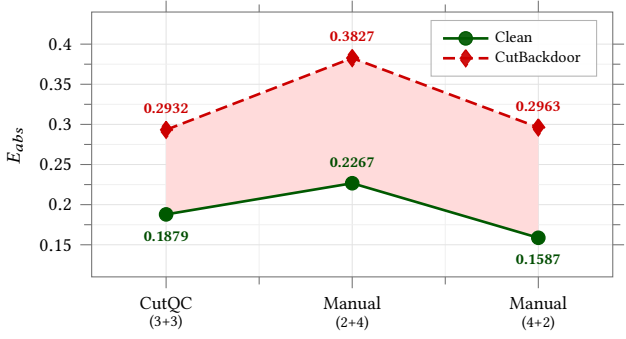
\subsubsection{Effect of different Cut locations.}
To assess whether CutBackdoor depends on a specific subcircuit  boundary, we evaluate three cut configurations on the H$_3^+$  VQE ansatz (6q), executed on IBMQ Manila (5q), where cutting is unavoidable. The first configuration uses CutQC's automatic MIP-based placement, producing a balanced (3+3) partition; the second and third are manually specified asymmetric partitions of (2+4) and (4+2) qubits. As shown in Figure~\ref{fig:cut_strategy}, CutBackdoor elevates the cut-path error consistently across all three configurations. The (2+4) partition shows the largest absolute gap, rising from 0.227 to 0.383; the automatic (3+3) partition rises from 0.188 to 0.293; and the (4+2) partition rises from 0.159 to 0.296. Across all three, the backdoor error is elevated by 56--87\% over the clean baseline, confirming that the attack does not depend on any specific subcircuit boundary. This follows directly from the structural gap established in \S\ref{sec:bound}, which is a property of the CutQC reconstruction pipeline itself rather than of how the circuit is partitioned.
\begin{figure}[t]
  \centering
  \definecolor{cleangreen}{RGB}{25,130,90}
\definecolor{backdoorred}{RGB}{200,75,40}
\definecolor{shadecolor}{RGB}{248,215,200}

\begin{tikzpicture}
\begin{axis}[
  width=8.0cm,
  height=6.0cm,
  xlabel={\footnotesize Backend},
  ylabel={\footnotesize $E_{abs}$ \,},
  xlabel style={font=\footnotesize, yshift=3pt},
  ylabel style={font=\footnotesize, yshift=-2pt},
  xtick={1,2,3,4},
  xticklabels={Kolkata, Manila, Lima, Quito},
  xticklabel style={font=\scriptsize},
  yticklabel style={font=\scriptsize},
  xmin=0.6, xmax=4.4,
  ymin=0.02, ymax=0.68,
  ytick={0.10,0.20,0.30,0.40,0.50,0.60},
  ymajorgrids=true,
  grid style={dotted, gray!40},
  axis line style={gray!60, thin},
  tick style={thin, gray!60},
  tick align=outside,
  legend style={
    font=\scriptsize,
    at={(0.04,0.98)},
    anchor=north west,
    draw=gray!50,
    fill=white,
    inner sep=3pt,
    row sep=2pt,
    legend cell align=left,
  },
  clip=false,
]

\addplot[name path=backdoor, color=backdoorred,
  mark=diamond*, mark size=2.8pt,
  mark options={fill=backdoorred},
  line width=1.1pt, dashed, dash pattern=on 3.5pt off 2pt,
] coordinates {
  (1, 0.2323)(2, 0.3390)(3, 0.2882)(4, 0.5588)
};
\addlegendentry{CutBackdoor }

\addplot[name path=clean, color=cleangreen,
  mark=*, mark size=2.2pt,
  mark options={fill=cleangreen},
  line width=1.1pt, solid,
] coordinates {
  (1, 0.1009)(2, 0.2161)(3, 0.1609)(4, 0.4154)
};
\addlegendentry{Clean}

\addplot[shadecolor, fill opacity=0.45, draw=none]
  fill between[of=backdoor and clean];

\node[font=\fontsize{5.2}{5.2}\selectfont\bfseries, color=cleangreen,
      anchor=north, yshift=-2pt] at (axis cs:1,0.1009) {0.1009};
\node[font=\fontsize{5.2}{5.2}\selectfont\bfseries, color=cleangreen,
      anchor=north, yshift=-2pt] at (axis cs:2,0.2161) {0.2161};
\node[font=\fontsize{5.2}{5.2}\selectfont\bfseries, color=cleangreen,
      anchor=north, yshift=-2pt] at (axis cs:3,0.1609) {0.1609};
\node[font=\fontsize{5.2}{5.2}\selectfont\bfseries, color=cleangreen,
      anchor=north, yshift=-2pt] at (axis cs:4,0.4154) {0.4154};

\node[font=\fontsize{5.2}{5.2}\selectfont\bfseries, color=backdoorred,
      anchor=south, yshift=2pt] at (axis cs:1,0.2323) {0.2323};
\node[font=\fontsize{5.2}{5.2}\selectfont\bfseries, color=backdoorred,
      anchor=south, yshift=2pt] at (axis cs:2,0.3390) {0.3390};
\node[font=\fontsize{5.2}{5.2}\selectfont\bfseries, color=backdoorred,
      anchor=south, yshift=2pt] at (axis cs:3,0.2882) {0.2882};
\node[font=\fontsize{5.2}{5.2}\selectfont\bfseries, color=backdoorred,
      anchor=south, yshift=2pt] at (axis cs:4,0.5588) {0.5588};

\end{axis}
\end{tikzpicture}   
  \caption{Effect of CutBackdoor with different Backend}
  \label{fig:backend}
\end{figure}
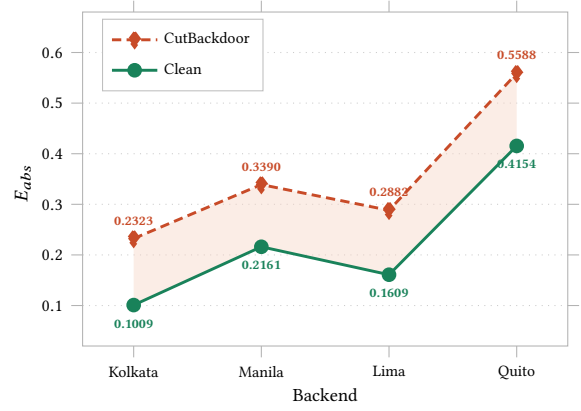
\subsubsection{Effect of Backend}
To evaluate whether CutBackdoor remains effective across hardware with varying noise characteristics, we execute the H$_3^+$ VQE ansatz with both clean and backdoor parameters through the CutQC pipeline on four IBMQ backends: IBMQ\_Kolkata (27 qubits), IBMQ\_Manila, IBMQ\_Lima, and IBMQ\_Quito (all 5-qubit backends), each exhibiting distinct gate error rates, readout errors, and coherence times. As shown in Figure~\ref{fig:backend}, the backdoor error is consistently elevated above the clean baseline across all four backends, with amplification ranging from $1.3\times$ on IBMQ\_Quito to $2.3\times$ on IBMQ\_Kolkata. The largest absolute backdoor error ($0.559$) occurs on IBMQ\_Quito, where the clean parameter error is also elevated ($0.415$), reflecting that IBMQ\_Quito exhibits the highest gate and readout error rates among the evaluated backends, and hardware noise degrades both clean and backdoor executions. Despite this elevated noise floor, CutBackdoor maintains a consistent gap above the clean baseline on every backend evaluated, confirming that the adversarial bias persists regardless of device noise level and that the attack is independent of device-specific noise characteristics.

\section{Defense}
\label{sec:defense}

Defending against backdoor attacks on VQAs remains an open problem, with no prior defense targeting the circuit-cutting execution context. Existing quantum backdoor defenses address either structurally detectable circuit modifications or device-specific noise conditions, neither of which applies to CutBackdoor. We adapt the classical fine-pruning framework~\cite{liu2018fine} to the VQA setting as a first candidate defense. The defender's only accessible signal is $\hat{E}_{\mathrm{cut}}$ itself, since full-circuit simulation is by assumption infeasible; the defense therefore operates exclusively through this observable, without access to clean reference parameters. We treat as a dormant unit any variational parameter whose gradient magnitude with respect to $\hat{E}_{\mathrm{cut}}$ is near zero at the compromised point $\boldsymbol{\theta}^*$, then perturb these dormant parameters and re-minimize $\hat{E}_{\mathrm{cut}}$ from $\boldsymbol{\theta}^*$.
Because the stealth constraint keeps $\boldsymbol{\theta}^*$ within the clean parameter basin by construction, this descent can move toward the legitimate variational minimum.

\begin{table}[t]
  \centering
  \caption{\revA{Fine-pruning defense: cut-path error $E_{\mathrm{abs}}$ (mHa) under backdoor parameters versus after prune-then-optimize on \texttt{IBMQ\_Kolkata}. Lower is better.}}
  \label{tab:defense}
  \begin{tabular}{lcc}
    \toprule
    Benchmark & CutBackdoor & Fine-pruning \\
    \midrule
    VQE -- H$_3^+$ & 199 & 119 \\
    VQD -- H$_3^+$ & 312 &  82 \\
    VQE -- CH$_2$  & 2695 & 2493 \\
    VQD -- H$_4$   & 597 & 492 \\
    \bottomrule
  \end{tabular}
\end{table}

\revA{We implement and evaluate this procedure on \texttt{IBMQ\_Kolkata} (Table~\ref{tab:defense}). Fine-pruning substantially reduces cut-path error on the shallow H$_3^+$ circuits (VQE $199\!\to\!119$, VQD $312\!\to\!82$\,mHa) but yields only marginal reductions on the deeper CH$_2$ and H$_4$ ansätze, where the residual error remains well above chemical accuracy. More importantly, the defense requires re-optimizing the parameters, which reintroduces exactly the training cost the victim downloaded shared parameters to avoid. A victim able to afford that re-optimization would not be in the vulnerable position the threat model assumes, so fine-pruning does not neutralize CutBackdoor under realistic constraints. Finally, an adaptive attacker who distributes adversarial bias across high-gradient parameters could evade gradient-based pruning entirely, leaving both stronger attacks and more robust defenses open.}

\section{Discussion and Future Work}
\label{sec:discussion}
While CutBackdoor establishes a concrete threat against CutQC-based VQA deployment, several directions remain open. \revC{We target CutQC as the first and most widely adopted automated wire-cutting pipeline, but the attack is not tied to it. The landscape now includes QPD-based gate cutting, randomized wire cutting, tomography-driven reconstruction, and classical shadow methods, each reconstructing the observable through a distinct pipeline. Because CutBackdoor is carried by the variational parameters rather than any CutQC-specific artifact, it generalizes by re-targeting the dual-objective loss to each procedure: QPD-based gate cutting reconstructs expectation values directly, and tomography-based methods access off-diagonal Pauli contributions and present a smaller attack surface. Each demands a reformulated objective.} Emerging frameworks are also moving toward adaptive strategies that re-partition circuits when output variance exceeds a threshold. Such pipelines may disrupt the fixed reconstruction structure CutBackdoor relies on, serving as a passive defense; conversely, an adversary anticipating this could craft poisoned parameters robust across a distribution of cut placements and hardware conditions. The interplay between adaptive cutting and adversarial robustness remains a promising open direction.

\section{Conclusion}
CutBackdoor is the first parameter-supply-chain backdoor that exploits
CutQC-based execution as an adversarial trigger against VQAs, requiring no attacker presence at deployment and no modification to the circuit architecture. Across the VQE and VQD benchmarks on multiple IBM quantum processors, CutBackdoor produces cut-path energy amplification of $1.3\times$ to $2.9\times$ over clean baselines, persisting across the evaluated backends and cut placements under matched compilation; Zero-Noise Extrapolation provides only partial mitigation, and the diagonal-cost QAOA benchmark marks the attack's structural boundary. As VQA deployments scale and parameter sharing grows, securing the circuit-cutting pipeline against supply-chain
threats is a timely challenge for the quantum security community.

\bibliographystyle{ACM-Reference-Format}
\bibliography{sample-base}


\begin{thebibliography}{69}


\ifx \showCODEN    \undefined \def \showCODEN     #1{\unskip}     \fi
\ifx \showISBNx    \undefined \def \showISBNx     #1{\unskip}     \fi
\ifx \showISBNxiii \undefined \def \showISBNxiii  #1{\unskip}     \fi
\ifx \showISSN     \undefined \def \showISSN      #1{\unskip}     \fi
\ifx \showLCCN     \undefined \def \showLCCN      #1{\unskip}     \fi
\ifx \shownote     \undefined \def \shownote      #1{#1}          \fi
\ifx \showarticletitle \undefined \def \showarticletitle #1{#1}   \fi
\ifx \showURL      \undefined \def \showURL       {\relax}        \fi
\providecommand\bibfield[2]{#2}
\providecommand\bibinfo[2]{#2}
\providecommand\natexlab[1]{#1}
\providecommand\showeprint[2][]{arXiv:#2}

\bibitem[Augustino et~al\mbox{.}(2024)]%
        {augustino2024strategies}
\bibfield{author}{\bibinfo{person}{Brandon Augustino}, \bibinfo{person}{Madelyn Cain}, \bibinfo{person}{Edward Farhi}, \bibinfo{person}{Swati Gupta}, \bibinfo{person}{Sam Gutmann}, \bibinfo{person}{Daniel Ranard}, \bibinfo{person}{Eugene Tang}, {and} \bibinfo{person}{Katherine Van~Kirk}.} \bibinfo{year}{2024}\natexlab{}.
\newblock \showarticletitle{Strategies for running the QAOA at hundreds of qubits}.
\newblock \bibinfo{journal}{\emph{arXiv preprint arXiv:2410.03015}} (\bibinfo{year}{2024}).
\newblock


\bibitem[Azad and Fomichev(2023)]%
        {azad2023pennylane}
\bibfield{author}{\bibinfo{person}{Utkarsh Azad} {and} \bibinfo{person}{Stepan Fomichev}.} \bibinfo{year}{2023}\natexlab{}.
\newblock \showarticletitle{Pennylane quantum chemistry datasets}.
\newblock \bibinfo{journal}{\emph{Accessed: Jul}}  \bibinfo{volume}{19} (\bibinfo{year}{2023}), \bibinfo{pages}{2025}.
\newblock


\bibitem[Bechtold et~al\mbox{.}(2023)]%
        {bechtold2023investigating}
\bibfield{author}{\bibinfo{person}{Marvin Bechtold}, \bibinfo{person}{Johanna Barzen}, \bibinfo{person}{Frank Leymann}, \bibinfo{person}{Alexander Mandl}, \bibinfo{person}{Julian Obst}, \bibinfo{person}{Felix Truger}, {and} \bibinfo{person}{Benjamin Weder}.} \bibinfo{year}{2023}\natexlab{}.
\newblock \showarticletitle{Investigating the effect of circuit cutting in QAOA for the MaxCut problem on NISQ devices}.
\newblock \bibinfo{journal}{\emph{Quantum Science and Technology}} \bibinfo{volume}{8}, \bibinfo{number}{4} (\bibinfo{year}{2023}), \bibinfo{pages}{045022}.
\newblock


\bibitem[Bergholm et~al\mbox{.}(2018)]%
        {bergholm2018pennylane}
\bibfield{author}{\bibinfo{person}{Ville Bergholm}, \bibinfo{person}{Josh Izaac}, \bibinfo{person}{Maria Schuld}, \bibinfo{person}{Christian Gogolin}, \bibinfo{person}{Shahnawaz Ahmed}, \bibinfo{person}{Vishnu Ajith}, \bibinfo{person}{M~Sohaib Alam}, \bibinfo{person}{Guillermo Alonso-Linaje}, \bibinfo{person}{Bharath AkashNarayanan}, \bibinfo{person}{Ali Asadi}, {et~al\mbox{.}}} \bibinfo{year}{2018}\natexlab{}.
\newblock \showarticletitle{Pennylane: Automatic differentiation of hybrid quantum-classical computations}.
\newblock \bibinfo{journal}{\emph{arXiv preprint arXiv:1811.04968}} (\bibinfo{year}{2018}).
\newblock


\bibitem[Carrera~Vazquez et~al\mbox{.}(2024)]%
        {CarreraVazquez2024}
\bibfield{author}{\bibinfo{person}{Almudena Carrera~Vazquez}, \bibinfo{person}{Caroline Tornow}, \bibinfo{person}{Diego Rist{\`e}}, \bibinfo{person}{Stefan Woerner}, \bibinfo{person}{Maika Takita}, {and} \bibinfo{person}{Daniel~J. Egger}.} \bibinfo{year}{2024}\natexlab{}.
\newblock \showarticletitle{Combining quantum processors with real-time classical communication}.
\newblock \bibinfo{journal}{\emph{Nature}} \bibinfo{volume}{636}, \bibinfo{number}{8041} (\bibinfo{date}{01 Dec} \bibinfo{year}{2024}), \bibinfo{pages}{75--79}.
\newblock
\showISSN{1476-4687}
\href{https://doi.org/10.1038/s41586-024-08178-2}{doi:\nolinkurl{10.1038/s41586-024-08178-2}}


\bibitem[Cerezo et~al\mbox{.}(2021)]%
        {cerezo2021variational}
\bibfield{author}{\bibinfo{person}{Marco Cerezo}, \bibinfo{person}{Andrew Arrasmith}, \bibinfo{person}{Ryan Babbush}, \bibinfo{person}{Simon~C Benjamin}, \bibinfo{person}{Suguru Endo}, \bibinfo{person}{Keisuke Fujii}, \bibinfo{person}{Jarrod~R McClean}, \bibinfo{person}{Kosuke Mitarai}, \bibinfo{person}{Xiao Yuan}, \bibinfo{person}{Lukasz Cincio}, {et~al\mbox{.}}} \bibinfo{year}{2021}\natexlab{}.
\newblock \showarticletitle{Variational quantum algorithms}.
\newblock \bibinfo{journal}{\emph{Nature Reviews Physics}} \bibinfo{volume}{3}, \bibinfo{number}{9} (\bibinfo{year}{2021}), \bibinfo{pages}{625--644}.
\newblock


\bibitem[Chong et~al\mbox{.}(2017)]%
        {chong2017programming}
\bibfield{author}{\bibinfo{person}{Frederic~T Chong}, \bibinfo{person}{Diana Franklin}, {and} \bibinfo{person}{Margaret Martonosi}.} \bibinfo{year}{2017}\natexlab{}.
\newblock \showarticletitle{Programming languages and compiler design for realistic quantum hardware}.
\newblock \bibinfo{journal}{\emph{Nature}} \bibinfo{volume}{549}, \bibinfo{number}{7671} (\bibinfo{year}{2017}), \bibinfo{pages}{180--187}.
\newblock


\bibitem[Chu et~al\mbox{.}(2023a)]%
        {chu2023qdoor}
\bibfield{author}{\bibinfo{person}{Cheng Chu}, \bibinfo{person}{Fan Chen}, \bibinfo{person}{Philip Richerme}, {and} \bibinfo{person}{Lei Jiang}.} \bibinfo{year}{2023}\natexlab{a}.
\newblock \showarticletitle{Qdoor: Exploiting approximate synthesis for backdoor attacks in quantum neural networks}. In \bibinfo{booktitle}{\emph{2023 IEEE International Conference on Quantum Computing and Engineering (QCE)}}, Vol.~\bibinfo{volume}{1}. IEEE, \bibinfo{pages}{1098--1106}.
\newblock


\bibitem[Chu et~al\mbox{.}(2025a)]%
        {chu2025lstm}
\bibfield{author}{\bibinfo{person}{Cheng Chu}, \bibinfo{person}{Aishwarya Hastak}, {and} \bibinfo{person}{Fan Chen}.} \bibinfo{year}{2025}\natexlab{a}.
\newblock \showarticletitle{Lstm-qgan: Scalable nisq generative adversarial network}. In \bibinfo{booktitle}{\emph{ICASSP 2025-2025 IEEE International Conference on Acoustics, Speech and Signal Processing (ICASSP)}}. IEEE, \bibinfo{pages}{1--5}.
\newblock


\bibitem[Chu et~al\mbox{.}(2023b)]%
        {chu2023qtrojan}
\bibfield{author}{\bibinfo{person}{Cheng Chu}, \bibinfo{person}{Lei Jiang}, \bibinfo{person}{Martin Swany}, {and} \bibinfo{person}{Fan Chen}.} \bibinfo{year}{2023}\natexlab{b}.
\newblock \showarticletitle{Qtrojan: A circuit backdoor against quantum neural networks}. In \bibinfo{booktitle}{\emph{ICASSP 2023-2023 IEEE International Conference on Acoustics, Speech and Signal Processing (ICASSP)}}. IEEE, \bibinfo{pages}{1--5}.
\newblock


\bibitem[Chu et~al\mbox{.}(2025b)]%
        {chuqnbad}
\bibfield{author}{\bibinfo{person}{Cheng Chu}, \bibinfo{person}{Qian Lou}, \bibinfo{person}{Fan Chen}, {and} \bibinfo{person}{Lei Jiang}.} \bibinfo{year}{2025}\natexlab{b}.
\newblock \showarticletitle{QNBAD: Quantum Noise-induced Backdoor Attacks against Zero Noise Extrapolation}.
\newblock  (\bibinfo{year}{2025}).
\newblock


\bibitem[Chu et~al\mbox{.}(2023c)]%
        {chu2023iqgan}
\bibfield{author}{\bibinfo{person}{Cheng Chu}, \bibinfo{person}{Grant Skipper}, \bibinfo{person}{Martin Swany}, {and} \bibinfo{person}{Fan Chen}.} \bibinfo{year}{2023}\natexlab{c}.
\newblock \showarticletitle{Iqgan: Robust quantum generative adversarial network for image synthesis on nisq devices}. In \bibinfo{booktitle}{\emph{ICASSP 2023-2023 IEEE international conference on acoustics, speech and signal processing (ICASSP)}}. IEEE, \bibinfo{pages}{1--5}.
\newblock


\bibitem[Crooks(2018)]%
        {crooks2018performance}
\bibfield{author}{\bibinfo{person}{Gavin~E Crooks}.} \bibinfo{year}{2018}\natexlab{}.
\newblock \showarticletitle{Performance of the quantum approximate optimization algorithm on the maximum cut problem}.
\newblock \bibinfo{journal}{\emph{arXiv preprint arXiv:1811.08419}} (\bibinfo{year}{2018}).
\newblock


\bibitem[Das and Ghosh(2023)]%
        {das2023randomized}
\bibfield{author}{\bibinfo{person}{Subrata Das} {and} \bibinfo{person}{Swaroop Ghosh}.} \bibinfo{year}{2023}\natexlab{}.
\newblock \showarticletitle{Randomized reversible gate-based obfuscation for secured compilation of quantum circuit}.
\newblock \bibinfo{journal}{\emph{arXiv preprint arXiv:2305.01133}} (\bibinfo{year}{2023}).
\newblock


\bibitem[Das and Ghosh(2024)]%
        {das2024trojan}
\bibfield{author}{\bibinfo{person}{Subrata Das} {and} \bibinfo{person}{Swaroop Ghosh}.} \bibinfo{year}{2024}\natexlab{}.
\newblock \showarticletitle{Trojan attacks on variational quantum circuits and countermeasures}. In \bibinfo{booktitle}{\emph{2024 25th International Symposium on Quality Electronic Design (ISQED)}}. IEEE, \bibinfo{pages}{1--8}.
\newblock


\bibitem[Egger et~al\mbox{.}(2020a)]%
        {egger2020quantum}
\bibfield{author}{\bibinfo{person}{Daniel~J Egger}, \bibinfo{person}{Claudio Gambella}, \bibinfo{person}{Jakub Marecek}, \bibinfo{person}{Scott McFaddin}, \bibinfo{person}{Martin Mevissen}, \bibinfo{person}{Rudy Raymond}, \bibinfo{person}{Andrea Simonetto}, \bibinfo{person}{Stefan Woerner}, {and} \bibinfo{person}{Elena Yndurain}.} \bibinfo{year}{2020}\natexlab{a}.
\newblock \showarticletitle{Quantum computing for finance: State-of-the-art and future prospects}.
\newblock \bibinfo{journal}{\emph{IEEE Transactions on Quantum Engineering}}  \bibinfo{volume}{1} (\bibinfo{year}{2020}), \bibinfo{pages}{1--24}.
\newblock


\bibitem[Egger et~al\mbox{.}(2020b)]%
        {egger2020credit}
\bibfield{author}{\bibinfo{person}{Daniel~J Egger}, \bibinfo{person}{Ricardo~Garc{\'\i}a Guti{\'e}rrez}, \bibinfo{person}{Jordi~Cahu{\'e} Mestre}, {and} \bibinfo{person}{Stefan Woerner}.} \bibinfo{year}{2020}\natexlab{b}.
\newblock \showarticletitle{Credit risk analysis using quantum computers}.
\newblock \bibinfo{journal}{\emph{IEEE transactions on computers}} \bibinfo{volume}{70}, \bibinfo{number}{12} (\bibinfo{year}{2020}), \bibinfo{pages}{2136--2145}.
\newblock


\bibitem[Farhi et~al\mbox{.}(2014)]%
        {farhi2014quantum}
\bibfield{author}{\bibinfo{person}{Edward Farhi}, \bibinfo{person}{Jeffrey Goldstone}, {and} \bibinfo{person}{Sam Gutmann}.} \bibinfo{year}{2014}\natexlab{}.
\newblock \showarticletitle{A quantum approximate optimization algorithm}.
\newblock \bibinfo{journal}{\emph{arXiv preprint arXiv:1411.4028}} (\bibinfo{year}{2014}).
\newblock


\bibitem[Farrell et~al\mbox{.}(2024)]%
        {farrell2024scalable}
\bibfield{author}{\bibinfo{person}{Roland~C Farrell}, \bibinfo{person}{Marc Illa}, \bibinfo{person}{Anthony~N Ciavarella}, {and} \bibinfo{person}{Martin~J Savage}.} \bibinfo{year}{2024}\natexlab{}.
\newblock \showarticletitle{Scalable circuits for preparing ground states on digital quantum computers: The Schwinger model vacuum on 100 qubits}.
\newblock \bibinfo{journal}{\emph{PRX Quantum}} \bibinfo{volume}{5}, \bibinfo{number}{2} (\bibinfo{year}{2024}), \bibinfo{pages}{020315}.
\newblock


\bibitem[Fradkin(1989)]%
        {fradkin1989jordan}
\bibfield{author}{\bibinfo{person}{Eduardo Fradkin}.} \bibinfo{year}{1989}\natexlab{}.
\newblock \showarticletitle{Jordan-Wigner transformation for quantum-spin systems in two dimensions and fractional statistics}.
\newblock \bibinfo{journal}{\emph{Physical review letters}} \bibinfo{volume}{63}, \bibinfo{number}{3} (\bibinfo{year}{1989}), \bibinfo{pages}{322}.
\newblock


\bibitem[Fujii et~al\mbox{.}(2022)]%
        {fujii2022deep}
\bibfield{author}{\bibinfo{person}{Keisuke Fujii}, \bibinfo{person}{Kaoru Mizuta}, \bibinfo{person}{Hiroshi Ueda}, \bibinfo{person}{Kosuke Mitarai}, \bibinfo{person}{Wataru Mizukami}, {and} \bibinfo{person}{Yuya~O Nakagawa}.} \bibinfo{year}{2022}\natexlab{}.
\newblock \showarticletitle{Deep variational quantum eigensolver: A divide-and-conquer method for solving a larger problem with smaller size quantum computers}.
\newblock \bibinfo{journal}{\emph{PRX Quantum}} \bibinfo{volume}{3}, \bibinfo{number}{1} (\bibinfo{year}{2022}), \bibinfo{pages}{010346}.
\newblock


\bibitem[Galda et~al\mbox{.}(2021)]%
        {galda2021transferability}
\bibfield{author}{\bibinfo{person}{Alexey Galda}, \bibinfo{person}{Xiaoyuan Liu}, \bibinfo{person}{Danylo Lykov}, \bibinfo{person}{Yuri Alexeev}, {and} \bibinfo{person}{Ilya Safro}.} \bibinfo{year}{2021}\natexlab{}.
\newblock \showarticletitle{Transferability of optimal QAOA parameters between random graphs}. In \bibinfo{booktitle}{\emph{2021 IEEE International Conference on Quantum Computing and Engineering (QCE)}}. IEEE, \bibinfo{pages}{171--180}.
\newblock


\bibitem[Gonthier et~al\mbox{.}(2022)]%
        {gonthier2022measurements}
\bibfield{author}{\bibinfo{person}{J{\'e}r{\^o}me~F Gonthier}, \bibinfo{person}{Maxwell~D Radin}, \bibinfo{person}{Corneliu Buda}, \bibinfo{person}{Eric~J Doskocil}, \bibinfo{person}{Clena~M Abuan}, {and} \bibinfo{person}{Jhonathan Romero}.} \bibinfo{year}{2022}\natexlab{}.
\newblock \showarticletitle{Measurements as a roadblock to near-term practical quantum advantage in chemistry: Resource analysis}.
\newblock \bibinfo{journal}{\emph{Physical Review Research}} \bibinfo{volume}{4}, \bibinfo{number}{3} (\bibinfo{year}{2022}), \bibinfo{pages}{033154}.
\newblock


\bibitem[Gu et~al\mbox{.}(2017)]%
        {gu2017badnets}
\bibfield{author}{\bibinfo{person}{Tianyu Gu}, \bibinfo{person}{Brendan Dolan-Gavitt}, {and} \bibinfo{person}{Siddharth Garg}.} \bibinfo{year}{2017}\natexlab{}.
\newblock \showarticletitle{Badnets: Identifying vulnerabilities in the machine learning model supply chain}.
\newblock \bibinfo{journal}{\emph{arXiv preprint arXiv:1708.06733}} (\bibinfo{year}{2017}).
\newblock


\bibitem[Guo et~al\mbox{.}(2025)]%
        {guo2025backdoor}
\bibfield{author}{\bibinfo{person}{Ji Guo}, \bibinfo{person}{Wenbo Jiang}, \bibinfo{person}{Rui Zhang}, \bibinfo{person}{Wenshu Fan}, \bibinfo{person}{Jiachen Li}, \bibinfo{person}{Guoming Lu}, {and} \bibinfo{person}{Hongwei Li}.} \bibinfo{year}{2025}\natexlab{}.
\newblock \showarticletitle{Backdoor attacks against hybrid classical-quantum neural networks}.
\newblock \bibinfo{journal}{\emph{Neural Networks}}  \bibinfo{volume}{191} (\bibinfo{year}{2025}), \bibinfo{pages}{107776}.
\newblock


\bibitem[{Gurobi Optimization, LLC}(2026)]%
        {gurobi2026}
\bibfield{author}{\bibinfo{person}{{Gurobi Optimization, LLC}}.} \bibinfo{year}{2026}\natexlab{}.
\newblock \bibinfo{booktitle}{\emph{Gurobi Optimizer Reference Manual}}.
\newblock
\newblock
\shownote{[Online]. Available: https://www.gurobi.com}.


\bibitem[Higgott et~al\mbox{.}(2019)]%
        {higgott2019variational}
\bibfield{author}{\bibinfo{person}{Oscar Higgott}, \bibinfo{person}{Daochen Wang}, {and} \bibinfo{person}{Stephen Brierley}.} \bibinfo{year}{2019}\natexlab{}.
\newblock \showarticletitle{Variational quantum computation of excited states}.
\newblock \bibinfo{journal}{\emph{Quantum}}  \bibinfo{volume}{3} (\bibinfo{year}{2019}), \bibinfo{pages}{156}.
\newblock


\bibitem[{IBM}(2022)]%
        {qiskit_cutting2022}
\bibfield{author}{\bibinfo{person}{{IBM}}.} \bibinfo{year}{2022}\natexlab{}.
\newblock \bibinfo{title}{Qiskit addon: circuit cutting}.
\newblock
\urldef\tempurl%
\url{https://github.com/Qiskit/qiskit-addon-cutting}
\showURL{%
\tempurl}
\newblock
\shownote{Accessed: 2026-03-30}.


\bibitem[Ithier et~al\mbox{.}(2005)]%
        {ithier2005decoherence}
\bibfield{author}{\bibinfo{person}{Gregoire Ithier}, \bibinfo{person}{E Collin}, \bibinfo{person}{P Joyez}, \bibinfo{person}{PJ Meeson}, \bibinfo{person}{Denis Vion}, \bibinfo{person}{Daniel Esteve}, \bibinfo{person}{F Chiarello}, \bibinfo{person}{A Shnirman}, \bibinfo{person}{Yu Makhlin}, \bibinfo{person}{Josef Schriefl}, {et~al\mbox{.}}} \bibinfo{year}{2005}\natexlab{}.
\newblock \showarticletitle{Decoherence in a superconducting quantum bit circuit}.
\newblock \bibinfo{journal}{\emph{Physical Review B—Condensed Matter and Materials Physics}} \bibinfo{volume}{72}, \bibinfo{number}{13} (\bibinfo{year}{2005}), \bibinfo{pages}{134519}.
\newblock


\bibitem[Javadi-Abhari et~al\mbox{.}(2024)]%
        {javadi2024quantum}
\bibfield{author}{\bibinfo{person}{Ali Javadi-Abhari}, \bibinfo{person}{Matthew Treinish}, \bibinfo{person}{Kevin Krsulich}, \bibinfo{person}{Christopher~J Wood}, \bibinfo{person}{Jake Lishman}, \bibinfo{person}{Julien Gacon}, \bibinfo{person}{Simon Martiel}, \bibinfo{person}{Paul~D Nation}, \bibinfo{person}{Lev~S Bishop}, \bibinfo{person}{Andrew~W Cross}, {et~al\mbox{.}}} \bibinfo{year}{2024}\natexlab{}.
\newblock \showarticletitle{Quantum computing with Qiskit}.
\newblock \bibinfo{journal}{\emph{arXiv preprint arXiv:2405.08810}} (\bibinfo{year}{2024}).
\newblock


\bibitem[Kandala et~al\mbox{.}(2017)]%
        {kandala2017hardware}
\bibfield{author}{\bibinfo{person}{Abhinav Kandala}, \bibinfo{person}{Antonio Mezzacapo}, \bibinfo{person}{Kristan Temme}, \bibinfo{person}{Maika Takita}, \bibinfo{person}{Markus Brink}, \bibinfo{person}{Jerry~M Chow}, {and} \bibinfo{person}{Jay~M Gambetta}.} \bibinfo{year}{2017}\natexlab{}.
\newblock \showarticletitle{Hardware-efficient variational quantum eigensolver for small molecules and quantum magnets}.
\newblock \bibinfo{journal}{\emph{nature}} \bibinfo{volume}{549}, \bibinfo{number}{7671} (\bibinfo{year}{2017}), \bibinfo{pages}{242--246}.
\newblock


\bibitem[Knill(2005)]%
        {knill2005quantum}
\bibfield{author}{\bibinfo{person}{Emanuel Knill}.} \bibinfo{year}{2005}\natexlab{}.
\newblock \showarticletitle{Quantum computing with realistically noisy devices}.
\newblock \bibinfo{journal}{\emph{Nature}} \bibinfo{volume}{434}, \bibinfo{number}{7029} (\bibinfo{year}{2005}), \bibinfo{pages}{39--44}.
\newblock


\bibitem[Larocca et~al\mbox{.}(2025)]%
        {larocca2025barren}
\bibfield{author}{\bibinfo{person}{Martin Larocca}, \bibinfo{person}{Supanut Thanasilp}, \bibinfo{person}{Samson Wang}, \bibinfo{person}{Kunal Sharma}, \bibinfo{person}{Jacob Biamonte}, \bibinfo{person}{Patrick~J Coles}, \bibinfo{person}{Lukasz Cincio}, \bibinfo{person}{Jarrod~R McClean}, \bibinfo{person}{Zo{\"e} Holmes}, {and} \bibinfo{person}{Marco Cerezo}.} \bibinfo{year}{2025}\natexlab{}.
\newblock \showarticletitle{Barren plateaus in variational quantum computing}.
\newblock \bibinfo{journal}{\emph{Nature Reviews Physics}} \bibinfo{volume}{7}, \bibinfo{number}{4} (\bibinfo{year}{2025}), \bibinfo{pages}{174--189}.
\newblock


\bibitem[LaRose et~al\mbox{.}(2022)]%
        {larose2022mitiq}
\bibfield{author}{\bibinfo{person}{Ryan LaRose}, \bibinfo{person}{Andrea Mari}, \bibinfo{person}{Sarah Kaiser}, \bibinfo{person}{Peter~J Karalekas}, \bibinfo{person}{Andre~A Alves}, \bibinfo{person}{Piotr Czarnik}, \bibinfo{person}{Mohamed El~Mandouh}, \bibinfo{person}{Max~H Gordon}, \bibinfo{person}{Yousef Hindy}, \bibinfo{person}{Aaron Robertson}, {et~al\mbox{.}}} \bibinfo{year}{2022}\natexlab{}.
\newblock \showarticletitle{Mitiq: A software package for error mitigation on noisy quantum computers}.
\newblock \bibinfo{journal}{\emph{Quantum}}  \bibinfo{volume}{6} (\bibinfo{year}{2022}), \bibinfo{pages}{774}.
\newblock


\bibitem[Li et~al\mbox{.}(2019)]%
        {li2019tackling}
\bibfield{author}{\bibinfo{person}{Gushu Li}, \bibinfo{person}{Yufei Ding}, {and} \bibinfo{person}{Yuan Xie}.} \bibinfo{year}{2019}\natexlab{}.
\newblock \showarticletitle{Tackling the qubit mapping problem for NISQ-era quantum devices}. In \bibinfo{booktitle}{\emph{Proceedings of the twenty-fourth international conference on architectural support for programming languages and operating systems}}. \bibinfo{pages}{1001--1014}.
\newblock


\bibitem[Li et~al\mbox{.}(2024)]%
        {li2024hybrid}
\bibfield{author}{\bibinfo{person}{Weitang Li}, \bibinfo{person}{Zhi Yin}, \bibinfo{person}{Xiaoran Li}, \bibinfo{person}{Dongqiang Ma}, \bibinfo{person}{Shuang Yi}, \bibinfo{person}{Zhenxing Zhang}, \bibinfo{person}{Chenji Zou}, \bibinfo{person}{Kunliang Bu}, \bibinfo{person}{Maochun Dai}, \bibinfo{person}{Jie Yue}, {et~al\mbox{.}}} \bibinfo{year}{2024}\natexlab{}.
\newblock \showarticletitle{A hybrid quantum computing pipeline for real world drug discovery}.
\newblock \bibinfo{journal}{\emph{Scientific Reports}} \bibinfo{volume}{14}, \bibinfo{number}{1} (\bibinfo{year}{2024}), \bibinfo{pages}{16942}.
\newblock


\bibitem[Lidar et~al\mbox{.}(1998)]%
        {lidar1998decoherence}
\bibfield{author}{\bibinfo{person}{Daniel~A Lidar}, \bibinfo{person}{Isaac~L Chuang}, {and} \bibinfo{person}{K~Birgitta Whaley}.} \bibinfo{year}{1998}\natexlab{}.
\newblock \showarticletitle{Decoherence free subspaces for quantum computation}.
\newblock \bibinfo{journal}{\emph{arXiv preprint quant-ph/9807004}} (\bibinfo{year}{1998}).
\newblock


\bibitem[Liu et~al\mbox{.}(2018)]%
        {liu2018fine}
\bibfield{author}{\bibinfo{person}{Kang Liu}, \bibinfo{person}{Brendan Dolan-Gavitt}, {and} \bibinfo{person}{Siddharth Garg}.} \bibinfo{year}{2018}\natexlab{}.
\newblock \showarticletitle{Fine-pruning: Defending against backdooring attacks on deep neural networks}. In \bibinfo{booktitle}{\emph{International symposium on research in attacks, intrusions, and defenses}}. Springer, \bibinfo{pages}{273--294}.
\newblock


\bibitem[Liu et~al\mbox{.}(2022)]%
        {9669165}
\bibfield{author}{\bibinfo{person}{Xiaoyuan Liu}, \bibinfo{person}{Anthony Angone}, \bibinfo{person}{Ruslan Shaydulin}, \bibinfo{person}{Ilya Safro}, \bibinfo{person}{Yuri Alexeev}, {and} \bibinfo{person}{Lukasz Cincio}.} \bibinfo{year}{2022}\natexlab{}.
\newblock \showarticletitle{Layer VQE: A Variational Approach for Combinatorial Optimization on Noisy Quantum Computers}.
\newblock \bibinfo{journal}{\emph{IEEE Transactions on Quantum Engineering}}  \bibinfo{volume}{3} (\bibinfo{year}{2022}), \bibinfo{pages}{1--20}.
\newblock
\href{https://doi.org/10.1109/TQE.2021.3140190}{doi:\nolinkurl{10.1109/TQE.2021.3140190}}


\bibitem[Lowe et~al\mbox{.}(2023)]%
        {lowe2023fast}
\bibfield{author}{\bibinfo{person}{Angus Lowe}, \bibinfo{person}{Matija Medvidovi{\'c}}, \bibinfo{person}{Anthony Hayes}, \bibinfo{person}{Lee~J O'Riordan}, \bibinfo{person}{Thomas~R Bromley}, \bibinfo{person}{Juan~Miguel Arrazola}, {and} \bibinfo{person}{Nathan Killoran}.} \bibinfo{year}{2023}\natexlab{}.
\newblock \showarticletitle{Fast quantum circuit cutting with randomized measurements}.
\newblock \bibinfo{journal}{\emph{Quantum}}  \bibinfo{volume}{7} (\bibinfo{year}{2023}), \bibinfo{pages}{934}.
\newblock


\bibitem[McArdle et~al\mbox{.}(2020)]%
        {mcardle2020quantum}
\bibfield{author}{\bibinfo{person}{Sam McArdle}, \bibinfo{person}{Suguru Endo}, \bibinfo{person}{Al{\'a}n Aspuru-Guzik}, \bibinfo{person}{Simon~C Benjamin}, {and} \bibinfo{person}{Xiao Yuan}.} \bibinfo{year}{2020}\natexlab{}.
\newblock \showarticletitle{Quantum computational chemistry}.
\newblock \bibinfo{journal}{\emph{Reviews of Modern Physics}} \bibinfo{volume}{92}, \bibinfo{number}{1} (\bibinfo{year}{2020}), \bibinfo{pages}{015003}.
\newblock


\bibitem[McClean et~al\mbox{.}(2018)]%
        {mcclean2018barren}
\bibfield{author}{\bibinfo{person}{Jarrod~R McClean}, \bibinfo{person}{Sergio Boixo}, \bibinfo{person}{Vadim~N Smelyanskiy}, \bibinfo{person}{Ryan Babbush}, {and} \bibinfo{person}{Hartmut Neven}.} \bibinfo{year}{2018}\natexlab{}.
\newblock \showarticletitle{Barren plateaus in quantum neural network training landscapes}.
\newblock \bibinfo{journal}{\emph{Nature communications}} \bibinfo{volume}{9}, \bibinfo{number}{1} (\bibinfo{year}{2018}), \bibinfo{pages}{4812}.
\newblock


\bibitem[McClean et~al\mbox{.}(2016)]%
        {mcclean2016theory}
\bibfield{author}{\bibinfo{person}{Jarrod~R McClean}, \bibinfo{person}{Jonathan Romero}, \bibinfo{person}{Ryan Babbush}, {and} \bibinfo{person}{Al{\'a}n Aspuru-Guzik}.} \bibinfo{year}{2016}\natexlab{}.
\newblock \showarticletitle{The theory of variational hybrid quantum-classical algorithms}.
\newblock \bibinfo{journal}{\emph{New Journal of Physics}} \bibinfo{volume}{18}, \bibinfo{number}{2} (\bibinfo{year}{2016}), \bibinfo{pages}{023023}.
\newblock


\bibitem[Mitarai and Fujii(2021)]%
        {mitarai2021constructing}
\bibfield{author}{\bibinfo{person}{Kosuke Mitarai} {and} \bibinfo{person}{Keisuke Fujii}.} \bibinfo{year}{2021}\natexlab{}.
\newblock \showarticletitle{Constructing a virtual two-qubit gate by sampling single-qubit operations}.
\newblock \bibinfo{journal}{\emph{New Journal of Physics}} \bibinfo{volume}{23}, \bibinfo{number}{2} (\bibinfo{year}{2021}), \bibinfo{pages}{023021}.
\newblock


\bibitem[Murali et~al\mbox{.}(2019)]%
        {murali2019noise}
\bibfield{author}{\bibinfo{person}{Prakash Murali}, \bibinfo{person}{Jonathan~M Baker}, \bibinfo{person}{Ali Javadi-Abhari}, \bibinfo{person}{Frederic~T Chong}, {and} \bibinfo{person}{Margaret Martonosi}.} \bibinfo{year}{2019}\natexlab{}.
\newblock \showarticletitle{Noise-adaptive compiler mappings for noisy intermediate-scale quantum computers}. In \bibinfo{booktitle}{\emph{Proceedings of the twenty-fourth international conference on architectural support for programming languages and operating systems}}. \bibinfo{pages}{1015--1029}.
\newblock


\bibitem[O’Malley et~al\mbox{.}(2016)]%
        {o2016scalable}
\bibfield{author}{\bibinfo{person}{Peter~JJ O’Malley}, \bibinfo{person}{Ryan Babbush}, \bibinfo{person}{Ian~D Kivlichan}, \bibinfo{person}{Jonathan Romero}, \bibinfo{person}{Jarrod~R McClean}, \bibinfo{person}{Rami Barends}, \bibinfo{person}{Julian Kelly}, \bibinfo{person}{Pedram Roushan}, \bibinfo{person}{Andrew Tranter}, \bibinfo{person}{Nan Ding}, {et~al\mbox{.}}} \bibinfo{year}{2016}\natexlab{}.
\newblock \showarticletitle{Scalable quantum simulation of molecular energies}.
\newblock \bibinfo{journal}{\emph{Physical Review X}} \bibinfo{volume}{6}, \bibinfo{number}{3} (\bibinfo{year}{2016}), \bibinfo{pages}{031007}.
\newblock


\bibitem[Patel et~al\mbox{.}(2022)]%
        {patel2022quest}
\bibfield{author}{\bibinfo{person}{Tirthak Patel}, \bibinfo{person}{Ed Younis}, \bibinfo{person}{Costin Iancu}, \bibinfo{person}{Wibe de Jong}, {and} \bibinfo{person}{Devesh Tiwari}.} \bibinfo{year}{2022}\natexlab{}.
\newblock \showarticletitle{Quest: systematically approximating quantum circuits for higher output fidelity}. In \bibinfo{booktitle}{\emph{Proceedings of the 27th ACM International Conference on Architectural Support for Programming Languages and Operating Systems}}. \bibinfo{pages}{514--528}.
\newblock


\bibitem[Peng et~al\mbox{.}(2020)]%
        {peng2020simulating}
\bibfield{author}{\bibinfo{person}{Tianyi Peng}, \bibinfo{person}{Aram~W Harrow}, \bibinfo{person}{Maris Ozols}, {and} \bibinfo{person}{Xiaodi Wu}.} \bibinfo{year}{2020}\natexlab{}.
\newblock \showarticletitle{Simulating large quantum circuits on a small quantum computer}.
\newblock \bibinfo{journal}{\emph{Physical review letters}} \bibinfo{volume}{125}, \bibinfo{number}{15} (\bibinfo{year}{2020}), \bibinfo{pages}{150504}.
\newblock


\bibitem[Perlin et~al\mbox{.}(2021)]%
        {perlin2021quantum}
\bibfield{author}{\bibinfo{person}{Michael~A Perlin}, \bibinfo{person}{Zain~H Saleem}, \bibinfo{person}{Martin Suchara}, {and} \bibinfo{person}{James~C Osborn}.} \bibinfo{year}{2021}\natexlab{}.
\newblock \showarticletitle{Quantum circuit cutting with maximum-likelihood tomography}.
\newblock \bibinfo{journal}{\emph{npj Quantum Information}} \bibinfo{volume}{7}, \bibinfo{number}{1} (\bibinfo{year}{2021}), \bibinfo{pages}{64}.
\newblock


\bibitem[Peruzzo et~al\mbox{.}(2014)]%
        {peruzzo2014variational}
\bibfield{author}{\bibinfo{person}{Alberto Peruzzo}, \bibinfo{person}{Jarrod McClean}, \bibinfo{person}{Peter Shadbolt}, \bibinfo{person}{Man-Hong Yung}, \bibinfo{person}{Xiao-Qi Zhou}, \bibinfo{person}{Peter~J Love}, \bibinfo{person}{Al{\'a}n Aspuru-Guzik}, {and} \bibinfo{person}{Jeremy~L O’brien}.} \bibinfo{year}{2014}\natexlab{}.
\newblock \showarticletitle{A variational eigenvalue solver on a photonic quantum processor}.
\newblock \bibinfo{journal}{\emph{Nature communications}} \bibinfo{volume}{5}, \bibinfo{number}{1} (\bibinfo{year}{2014}), \bibinfo{pages}{4213}.
\newblock


\bibitem[Piveteau and Sutter(2023)]%
        {piveteau2023circuit}
\bibfield{author}{\bibinfo{person}{Christophe Piveteau} {and} \bibinfo{person}{David Sutter}.} \bibinfo{year}{2023}\natexlab{}.
\newblock \showarticletitle{Circuit knitting with classical communication}.
\newblock \bibinfo{journal}{\emph{IEEE Transactions on Information Theory}} \bibinfo{volume}{70}, \bibinfo{number}{4} (\bibinfo{year}{2023}), \bibinfo{pages}{2734--2745}.
\newblock


\bibitem[Preskill(2018)]%
        {preskill2018quantum}
\bibfield{author}{\bibinfo{person}{John Preskill}.} \bibinfo{year}{2018}\natexlab{}.
\newblock \showarticletitle{Quantum computing in the NISQ era and beyond}.
\newblock \bibinfo{journal}{\emph{Quantum}}  \bibinfo{volume}{2} (\bibinfo{year}{2018}), \bibinfo{pages}{79}.
\newblock


\bibitem[Romero et~al\mbox{.}(2019)]%
        {romero2019strategies}
\bibfield{author}{\bibinfo{person}{Jonathan Romero}, \bibinfo{person}{Ryan Babbush}, \bibinfo{person}{Jarrod~R McClean}, \bibinfo{person}{Cornelius Hempel}, \bibinfo{person}{Peter~J Love}, {and} \bibinfo{person}{Al{\'a}n Aspuru-Guzik}.} \bibinfo{year}{2019}\natexlab{}.
\newblock \showarticletitle{Strategies for quantum computing molecular energies using the unitary coupled cluster ansatz}.
\newblock \bibinfo{journal}{\emph{Quantum Science and Technology}} \bibinfo{volume}{4}, \bibinfo{number}{1} (\bibinfo{year}{2019}), \bibinfo{pages}{014008}.
\newblock


\bibitem[Sarovar et~al\mbox{.}(2020)]%
        {sarovar2020detecting}
\bibfield{author}{\bibinfo{person}{Mohan Sarovar}, \bibinfo{person}{Timothy Proctor}, \bibinfo{person}{Kenneth Rudinger}, \bibinfo{person}{Kevin Young}, \bibinfo{person}{Erik Nielsen}, {and} \bibinfo{person}{Robin Blume-Kohout}.} \bibinfo{year}{2020}\natexlab{}.
\newblock \showarticletitle{Detecting crosstalk errors in quantum information processors}.
\newblock \bibinfo{journal}{\emph{Quantum}}  \bibinfo{volume}{4} (\bibinfo{year}{2020}), \bibinfo{pages}{321}.
\newblock


\bibitem[Sawaya et~al\mbox{.}(2024)]%
        {sawaya2024hamlib}
\bibfield{author}{\bibinfo{person}{Nicolas~PD Sawaya}, \bibinfo{person}{Daniel Marti-Dafcik}, \bibinfo{person}{Yang Ho}, \bibinfo{person}{Daniel~P Tabor}, \bibinfo{person}{David E~Bernal Neira}, \bibinfo{person}{Alicia~B Magann}, \bibinfo{person}{Shavindra Premaratne}, \bibinfo{person}{Pradeep Dubey}, \bibinfo{person}{Anne Matsuura}, \bibinfo{person}{Nathan Bishop}, {et~al\mbox{.}}} \bibinfo{year}{2024}\natexlab{}.
\newblock \showarticletitle{HamLib: A library of Hamiltonians for benchmarking quantum algorithms and hardware}.
\newblock \bibinfo{journal}{\emph{Quantum}}  \bibinfo{volume}{8} (\bibinfo{year}{2024}), \bibinfo{pages}{1559}.
\newblock


\bibitem[Schuld and Petruccione(2018)]%
        {schuld2018supervised}
\bibfield{author}{\bibinfo{person}{Maria Schuld} {and} \bibinfo{person}{Francesco Petruccione}.} \bibinfo{year}{2018}\natexlab{}.
\newblock \bibinfo{booktitle}{\emph{Supervised learning with quantum computers}}. Vol.~\bibinfo{volume}{17}.
\newblock \bibinfo{publisher}{Springer}.
\newblock


\bibitem[Shaydulin et~al\mbox{.}(2023)]%
        {shaydulin2023parameter}
\bibfield{author}{\bibinfo{person}{Ruslan Shaydulin}, \bibinfo{person}{Phillip~C Lotshaw}, \bibinfo{person}{Jeffrey Larson}, \bibinfo{person}{James Ostrowski}, {and} \bibinfo{person}{Travis~S Humble}.} \bibinfo{year}{2023}\natexlab{}.
\newblock \showarticletitle{Parameter transfer for quantum approximate optimization of weighted maxcut}.
\newblock \bibinfo{journal}{\emph{ACM Transactions on Quantum Computing}} \bibinfo{volume}{4}, \bibinfo{number}{3} (\bibinfo{year}{2023}), \bibinfo{pages}{1--15}.
\newblock


\bibitem[Shaydulin et~al\mbox{.}(2021)]%
        {shaydulin2021qaoakit}
\bibfield{author}{\bibinfo{person}{Ruslan Shaydulin}, \bibinfo{person}{Kunal Marwaha}, \bibinfo{person}{Jonathan Wurtz}, {and} \bibinfo{person}{Phillip~C Lotshaw}.} \bibinfo{year}{2021}\natexlab{}.
\newblock \showarticletitle{QAOAKit: A toolkit for reproducible study, application, and verification of the QAOA}. In \bibinfo{booktitle}{\emph{2021 IEEE/ACM Second International Workshop on Quantum Computing Software (QCS)}}. IEEE, \bibinfo{pages}{64--71}.
\newblock


\bibitem[Shaydulin et~al\mbox{.}(2019)]%
        {shaydulin2019multistart}
\bibfield{author}{\bibinfo{person}{Ruslan Shaydulin}, \bibinfo{person}{Ilya Safro}, {and} \bibinfo{person}{Jeffrey Larson}.} \bibinfo{year}{2019}\natexlab{}.
\newblock \showarticletitle{Multistart methods for quantum approximate optimization}. In \bibinfo{booktitle}{\emph{2019 IEEE high performance extreme computing conference (HPEC)}}. IEEE, \bibinfo{pages}{1--8}.
\newblock


\bibitem[Skogh et~al\mbox{.}(2023)]%
        {skogh2023accelerating}
\bibfield{author}{\bibinfo{person}{M{\aa}rten Skogh}, \bibinfo{person}{Oskar Leinonen}, \bibinfo{person}{Phalgun Lolur}, {and} \bibinfo{person}{Martin Rahm}.} \bibinfo{year}{2023}\natexlab{}.
\newblock \showarticletitle{Accelerating variational quantum eigensolver convergence using parameter transfer}.
\newblock \bibinfo{journal}{\emph{Electronic Structure}} \bibinfo{volume}{5}, \bibinfo{number}{3} (\bibinfo{year}{2023}), \bibinfo{pages}{035002}.
\newblock


\bibitem[Smith et~al\mbox{.}(2025)]%
        {smith2025single}
\bibfield{author}{\bibinfo{person}{Molly~C Smith}, \bibinfo{person}{Aaron~D Leu}, \bibinfo{person}{Koichiro Miyanishi}, \bibinfo{person}{Mario~F Gely}, {and} \bibinfo{person}{David~M Lucas}.} \bibinfo{year}{2025}\natexlab{}.
\newblock \showarticletitle{Single-qubit gates with errors at the 10-7 level}.
\newblock \bibinfo{journal}{\emph{Physical Review Letters}} \bibinfo{volume}{134}, \bibinfo{number}{23} (\bibinfo{year}{2025}), \bibinfo{pages}{230601}.
\newblock


\bibitem[Sureshbabu et~al\mbox{.}(2024)]%
        {sureshbabu2024parameter}
\bibfield{author}{\bibinfo{person}{Shree~Hari Sureshbabu}, \bibinfo{person}{Dylan Herman}, \bibinfo{person}{Ruslan Shaydulin}, \bibinfo{person}{Joao Basso}, \bibinfo{person}{Shouvanik Chakrabarti}, \bibinfo{person}{Yue Sun}, {and} \bibinfo{person}{Marco Pistoia}.} \bibinfo{year}{2024}\natexlab{}.
\newblock \showarticletitle{Parameter setting in quantum approximate optimization of weighted problems}.
\newblock \bibinfo{journal}{\emph{Quantum}}  \bibinfo{volume}{8} (\bibinfo{year}{2024}), \bibinfo{pages}{1231}.
\newblock


\bibitem[Tang et~al\mbox{.}(2021)]%
        {tang2021cutqc}
\bibfield{author}{\bibinfo{person}{Wei Tang}, \bibinfo{person}{Teague Tomesh}, \bibinfo{person}{Martin Suchara}, \bibinfo{person}{Jeffrey Larson}, {and} \bibinfo{person}{Margaret Martonosi}.} \bibinfo{year}{2021}\natexlab{}.
\newblock \showarticletitle{Cutqc: using small quantum computers for large quantum circuit evaluations}. In \bibinfo{booktitle}{\emph{Proceedings of the 26th ACM International conference on architectural support for programming languages and operating systems}}. \bibinfo{pages}{473--486}.
\newblock


\bibitem[Temme et~al\mbox{.}(2017)]%
        {temme2017error}
\bibfield{author}{\bibinfo{person}{Kristan Temme}, \bibinfo{person}{Sergey Bravyi}, {and} \bibinfo{person}{Jay~M Gambetta}.} \bibinfo{year}{2017}\natexlab{}.
\newblock \showarticletitle{Error mitigation for short-depth quantum circuits}.
\newblock \bibinfo{journal}{\emph{Physical review letters}} \bibinfo{volume}{119}, \bibinfo{number}{18} (\bibinfo{year}{2017}), \bibinfo{pages}{180509}.
\newblock


\bibitem[Tilly et~al\mbox{.}(2022)]%
        {tilly2022variational}
\bibfield{author}{\bibinfo{person}{Jules Tilly}, \bibinfo{person}{Hongxiang Chen}, \bibinfo{person}{Shuxiang Cao}, \bibinfo{person}{Dario Picozzi}, \bibinfo{person}{Kanav Setia}, \bibinfo{person}{Ying Li}, \bibinfo{person}{Edward Grant}, \bibinfo{person}{Leonard Wossnig}, \bibinfo{person}{Ivan Rungger}, \bibinfo{person}{George~H Booth}, {et~al\mbox{.}}} \bibinfo{year}{2022}\natexlab{}.
\newblock \showarticletitle{The variational quantum eigensolver: a review of methods and best practices}.
\newblock \bibinfo{journal}{\emph{Physics Reports}}  \bibinfo{volume}{986} (\bibinfo{year}{2022}), \bibinfo{pages}{1--128}.
\newblock


\bibitem[Typaldos et~al\mbox{.}(2024)]%
        {typaldos2024leveraging}
\bibfield{author}{\bibinfo{person}{George Typaldos}, \bibinfo{person}{Wei Tang}, {and} \bibinfo{person}{Jakub Szefer}.} \bibinfo{year}{2024}\natexlab{}.
\newblock \showarticletitle{Leveraging quantum circuit cutting for obfuscation and intellectual property protection}. In \bibinfo{booktitle}{\emph{2024 IEEE International Conference on Quantum Computing and Engineering (QCE)}}, Vol.~\bibinfo{volume}{1}. IEEE, \bibinfo{pages}{1824--1834}.
\newblock


\bibitem[Typaldos et~al\mbox{.}(2025)]%
        {typaldos2025quantum}
\bibfield{author}{\bibinfo{person}{George Typaldos}, \bibinfo{person}{Theodoros Trochatos}, {and} \bibinfo{person}{Jakub Szefer}.} \bibinfo{year}{2025}\natexlab{}.
\newblock \showarticletitle{Quantum Circuit Cutting: A Security Methodology}. In \bibinfo{booktitle}{\emph{2025 IEEE International Conference on Quantum Computing and Engineering (QCE)}}, Vol.~\bibinfo{volume}{1}. IEEE, \bibinfo{pages}{417--427}.
\newblock


\bibitem[Wang et~al\mbox{.}(2021)]%
        {wang2021noise}
\bibfield{author}{\bibinfo{person}{Samson Wang}, \bibinfo{person}{Enrico Fontana}, \bibinfo{person}{Marco Cerezo}, \bibinfo{person}{Kunal Sharma}, \bibinfo{person}{Akira Sone}, \bibinfo{person}{Lukasz Cincio}, {and} \bibinfo{person}{Patrick~J Coles}.} \bibinfo{year}{2021}\natexlab{}.
\newblock \showarticletitle{Noise-induced barren plateaus in variational quantum algorithms}.
\newblock \bibinfo{journal}{\emph{Nature communications}} \bibinfo{volume}{12}, \bibinfo{number}{1} (\bibinfo{year}{2021}), \bibinfo{pages}{6961}.
\newblock


\bibitem[Wu et~al\mbox{.}(2021)]%
        {wu2021towards}
\bibfield{author}{\bibinfo{person}{Anbang Wu}, \bibinfo{person}{Gushu Li}, \bibinfo{person}{Yuke Wang}, \bibinfo{person}{Boyuan Feng}, \bibinfo{person}{Yufei Ding}, {and} \bibinfo{person}{Yuan Xie}.} \bibinfo{year}{2021}\natexlab{}.
\newblock \showarticletitle{Towards efficient ansatz architecture for variational quantum algorithms}.
\newblock \bibinfo{journal}{\emph{arXiv preprint arXiv:2111.13730}} (\bibinfo{year}{2021}).
\newblock


\end{thebibliography}

\end{document}